\documentclass[letterpaper,11pt]{article}

\usepackage[margin=1in]{geometry}  

\usepackage[english]{babel}
\usepackage[utf8x]{inputenc}
\usepackage[T1]{fontenc}
\usepackage{graphicx}
\usepackage{subcaption}
\usepackage{longtable} 
\usepackage{multirow}
\usepackage{subcaption}
\usepackage{listings}
\usepackage{makecell}
\usepackage{longtable}

\usepackage{amsmath}
\usepackage{amssymb}
\usepackage{amsthm}
\usepackage{graphicx}
\usepackage{natbib}
\usepackage[ruled]{algorithm2e}
\usepackage{longtable}
\usepackage[colorinlistoftodos]{todonotes}
\usepackage[
CJKbookmarks=true,
bookmarksnumbered=true,
bookmarksopen=true,
citecolor=red,
linkcolor=blue,
anchorcolor=red,
urlcolor=blue
]{hyperref}
\usepackage{enumitem}
\usepackage{mathtools}
\usepackage{mathrsfs}
\mathtoolsset{showonlyrefs}
\usepackage{tikz}
\usepackage{pgf}

\theoremstyle{plain}

\newtheorem{theorem}{Theorem}

\newtheorem{property}{Property}

\newtheorem{proposition}{Proposition}

\newcommand{\R}{\mathbb{R}}
\newcommand{\E}{\mathbb{E}}
\newcommand{\p}{\mathbb{P}}

\newcommand{\cov}{\mathrm{Cov}}

\newcommand{\indep}{\rotatebox[origin=c]{90}{$\models$}}
\newcommand{\sign}{\mathrm{sign}}
\newcommand{\Bern}{\text{Bern}}

\newcommand{\fdr}{\mathrm{FDR}}
\newcommand{\indc}[1]{{\mathbf{1}_{\left\{{#1}\right\}}}}

\newcommand{\argmax}[1]{\underset{#1}{\arg\!\max}}
\newcommand{\argmin}[1]{\underset{#1}{\arg\!\min}}

\usepackage{xspace}

\definecolor{myblue}{rgb}{.8, .8, 1}
\definecolor{mathblue}{rgb}{0.2472, 0.24, 0.6} 
\definecolor{mathred}{rgb}{0.6, 0.24, 0.442893}
\definecolor{mathyellow}{rgb}{0.6, 0.547014, 0.24}

\newcommand{\tX}{{\tilde{X}}}

\newcommand{\tZ}{{\tilde{Z}}}

\newcommand{\calB}{{\mathcal{B}}}

\newcommand{\calF}{{\mathcal{F}}}
\newcommand{\calG}{{\mathcal{G}}}
\newcommand{\calH}{{\mathcal{H}}}

\newcommand{\calN}{{\mathcal{N}}}

\newcommand{\calP}{{\mathcal{P}}}

\newcommand{\calS}{{\mathcal{S}}}

\newcommand{\eqd}{\stackrel{\textnormal{d}}{=}}
\newcommand{\reveal}{\mathrm{reveal}}
\long\def\comment#1{}

\newcommand{\overbar}[1]{\mkern 1.5mu\overline{\mkern-1.5mu#1\mkern-1.5mu}\mkern 1.5mu}

\usepackage{authblk}
\title{Knockoffs with Side Information}
\author[1]{Zhimei Ren}
\author[1,2]{Emmanuel Cand\`es}
\date{\today}
\affil[1]{Department of Statistics, Stanford University, Stanford, CA 94305}
\affil[2]{Department of Mathematics, Stanford University, Stanford, CA 94305}

\begin{document}

\maketitle


		

			
			
				
			
			
\begin{abstract}
  We consider the problem of assessing the importance of multiple
  variables or factors from a dataset when side information is
  available.  In principle, using side information can allow the
  statistician to pay attention to variables with a greater potential,
  which in turn, may lead to more discoveries. We introduce an
  adaptive knockoff filter, which generalizes the knockoff procedure
  \citep{barber2015controlling,candes2018panning} in that it uses both
  the data at hand and side information to adaptively order the
  variables under study and focus on those that are most
  promising. {\em Adaptive knockoffs} controls the finite-sample false
  discovery rate (FDR) and we demonstrate its power by comparing it
  with other structured multiple testing methods. We also apply our
  methodology to real genetic data in order to find associations
  between genetic variants and various phenotypes such as Crohn's
  disease and lipid levels. Here, adaptive knockoffs makes more
  discoveries than reported in previous studies on the same datasets.
\end{abstract}
\bigskip
{\small \textbf{Keywords.} Multiple testing, variable selection, false
  discovery rate (FDR), knockoff filters, Bayesian two-group model, 
  genome-wide association study (GWAS).}

\section{Introduction}\label{sec:intro}

Imagine a geneticist has collected genotype and phenotype data from a
population of individuals. She plans to use her data to study the
effect of genetic variants on a certain complex disease within this
population. Prior to data analysis, it is often the case that some
knowledge about the genetic variants under study is available: for
instance, there may be existing works on related diseases, as well as
research about the exact same disease and its occurrence within other
populations. How then should our geneticist leverage this prior
information in her own study?  Moving away from genetics, we broadly
recognize that researchers have more often than not access to prior
domain knowledge, results from relevant studies, and so on. Therefore,
the general question is this: how should they use side information in
their data analysis to help them discover more relevant factors?  How
should this be done while controlling type-I errors so that we do not
run into the problem of irreproducibility?  Our paper is motivated by
such common situations and objectives.
		
\subsection{Controlled variable selection methods}
We begin by formalizing the variable selection problem in statistical
terms. Let $X = (X_1,\ldots,X_p)$ denote the covariate vector and $Y$
the response variable. We assume that the pair $(X,Y)$ is sampled from
$P_X \cdot P_{Y|X} $, where $P_X$ is the marginal distribution of $X$
and $P_{Y|X}$ the conditional distribution of $Y|X$. The inferential
goal is to test whether this conditional distribution depends on $X_j$
or not.  We call feature $j$ a null if the conditional distribution of
$Y|X$ does not depend on $X_j$ and a non-null otherwise. With this in
mind, let $\calH_0$
denote the set of nulls and put
$\calH_1 = \{1,...,p\}\backslash\calH_0$. A controlled variable
selection method aims to detect non-nulls from a pool of candidates
while controlling some form of type-I error. In this paper, we
consider the false discovery rate (FDR)
\citep{benjamini1995controlling},
\begin{align}\label{fdr}
  \fdr = \E\,\, \left[ \dfrac{|\hat{\calS}\cap \calH_0|}{1 \vee
  	|\hat{\calS}|} \right] , 
\end{align} 
where $a\vee b =\max(a,b)$. Above,
$\hat{\calS} \subset \{1, \ldots, p\}$ is the selected set of
covariates and $|\cdot|$ is the cardinality of a set.
		
%
Most classical FDR-controlling procedures require that we have
available valid p-values, and further require independence or
constrained dependence between these p-values (e.g.,
\citet{benjamini1995controlling,benjamini2001control,storey2002direct,storey2004strong}). However,
it is in general challenging to obtain valid p-values for hypotheses
of interest, especially in the high-dimensional regime where the
sample size $n$ is on the order of the number $p$ of covariates or
less.  This is the reason why common practice usually imposes
stringent model assumptions and the validity of the p-values ends up
relying on the correctness of the model. Researchers have noted that
in common regimes, the p-values obtained by classical methods do not
behave as desired, but rather in a way that will potentially inflate
the FDR (see e.g.,
\citet{dezeure2015high,sur2017likelihood,sur2019modern}). \emph{Model-X
  knockoffs}, introduced in \citet{candes2018panning}, bypasses the
need for p-values and offers a solution to the variable selection
problem without making any modeling assumptions about the conditional
distribution of $Y|X$. The strength of this approach is that it does
not ask the statistician to assume away the form of the relationship
between the response variable and the family of covariates, namely,
$P_{Y|X}$ which is 1) usually unknown and 2) the actual object of
inference \citep{janson2017model}.  For instance, model-X knockoffs
does not ask the statistician to write down a convenient linear model
or a generalized linear model---which may or may not hold at all---to
describe the relationship between $X$ and $Y$.

This paper builds upon knockoffs and generalizes it to a setting where
side information about the variables or factors under study happens to
be available.
		
\subsection{Related works}
Previous works on multiple testing with side information broadly fall
into two categories. The first essentially modifies the definition of
the FDR to account for what is known. For example, we can use side
information to weigh each hypothese---e.g.~such that a priori
promising hypotheses receive a higher weight---and thereafter consider
controlling a weighted version of the FDR instead of the original FDR
(see e.g.,
\citet{benjamini1997multiple,benjamini2007false,basu2018weighted}). The
other category of works keeps the original FDR as a target measure and
aims at using side information to improve the power of the selection
procedure. Such procedures are sometimes called structured multiple
testing procedures and the line of work includes
\citet{genovese2006false,ferkingstad2008unsupervised,roeder2009genome,ignatiadis2016data,lei2016power,lynch2017control,ignatiadis2017covariate,lei2018adapt,li2019multiple};
and \citet{tony2019covariate}, among others.
		
In this paper, we adopt the second perspective. Our work is most
notably inspired by AdaPT of \cite{lei2018adapt} in that we
incorporate the idea of adaptively using side information within the
knockoffs framework. In a nutshell, AdaPT assumes that we can compute
independent p-values, which are then compared against a sequence of
{\em adaptive} thresholds constructed using available side
information. A clever calculation then produces estimates of the FDR
if the analyst were to report those hypotheses below threshold. (The
procedure iteratively lowers these thresholds until the FDR estimate
is below a target.)  In this paper, we work with model-X knockoffs,
which is completely different, and use side information to adaptively
screen knockoff importance statistics instead. An appealing feature is
that FDR control is achieved under the same conditions as for
(vanilla) model-X knockoffs: we (only) ask for the knowledge of the
distribution $P_X$ of the covariates, which is resasonable in many
situations \citep{candes2018panning}.\footnote{It is worth mentioning
  that \citet{lei2018adapt} discuss in passing (Section 6.2) removing
  the independence assumption by constructing knockoff copies of the
  p-values. This is different from our methodology, since
  in order to construct knockoff copies, one would need to know the
  joint distribution of the p-values, which is quite restrictive in
  practice.}

		
The reader will correctly note that the role of side information in
our framework is similar to that of a prior in the Bayesian
framework. However, our perspective on side information is here
frequentist and, therefore, intrinsically different. Bayesian
inference is obtained by averaging over the prior distribution and the
validity of inference relies on the correctness of the prior (and the
Bayesian model). In our work, the inference results (e.g., FDR
control, statistical power) hold conditional on the side information
and most importantly, the correctness of the side information does not
affect the validity of inference. Having said this, we shall see that
our adaptive knockoff filter accomodates `Bayesian thinking' in the
sense that side information can be assimilated into a prior, which can
then be used by our method while retaining control of the FDR. This
type-1 error guarantee holds no matter the validity of the prior or
the quality of side information.

\section{A motivating example: discovering genes with side
  information}\label{sec:motivatingexample}

In a nutshell, the vanilla knockoffs procedure
\citep{barber2015controlling,candes2018panning} uses the data at hand
to construct negative controls, which are then used to rank hypotheses
from the least to most promising. Selection is then achieved by
applying a special step-up procedure to these ranked hypotheses; see
Figure \ref{fig:ordering_intro} for a visual illustration. Having
ordered the hypotheses, the knockoff filter sequentially examines the
hypotheses starting with the least promising (i.e.~starting from the
left on the figure). As in the Benjamini-Hochberg step-up procedure
\citep{storey2004strong}, at each step, the knockoffs filter estimates
the FDR among the unexamined hypotheses. This estimate is the ratio
between the number of remaining non-candidate hypotheses and that of
remaining candidates. If the estimated FDR falls below a
user-specified threshold $q$, the procedure stops and selects the
remaining candidates (we will see later how knockoffs classify
hypotheses as candidates or not). Clearly, a greater number of
candidates at the end of the ordering yields higher power. The catch
however is this: a crucial rule for FDR control is that we are not
allowed to use the status of any hypothesis---whether it is a
candidate or not---when determining the ordering of the hypotheses (as
we would otherwise put all the candidates at the end).  Now suppose we
have side information other than the data itself. If we can use it to
come up with a better ordering and place more candidates towards the
end, then we will have a chance to select more hypotheses and,
therefore, improve power.

\begin{figure}[ht]
  \begin{subfigure}{0.49\textwidth}
    \centering \includegraphics[width =
    1\textwidth]{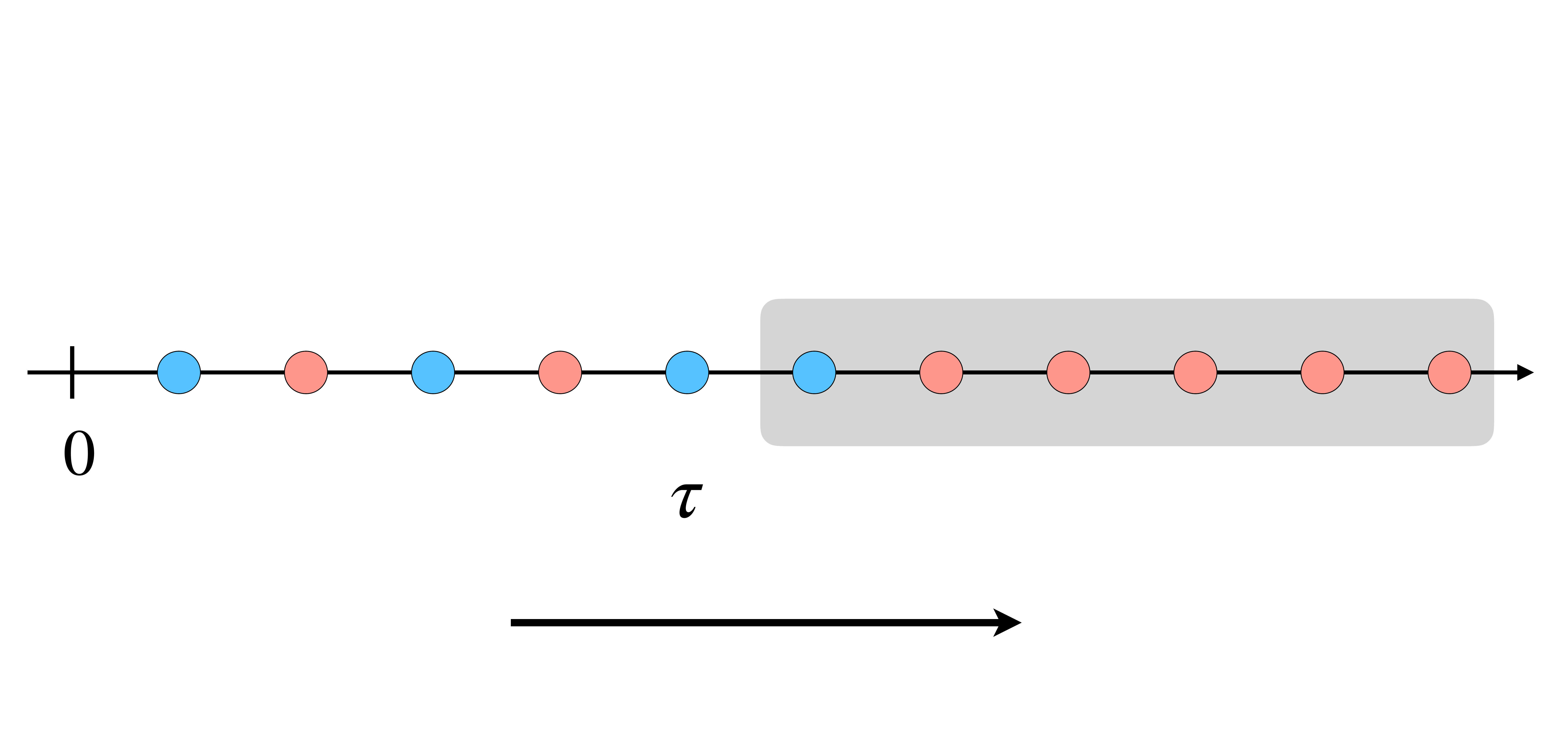}
    \caption{``Worse'' ordering.}
    \label{fig:intro_badord}
  \end{subfigure}
  \hfill
  \begin{subfigure}{0.49\textwidth}
    \centering \includegraphics[width =
    1\textwidth]{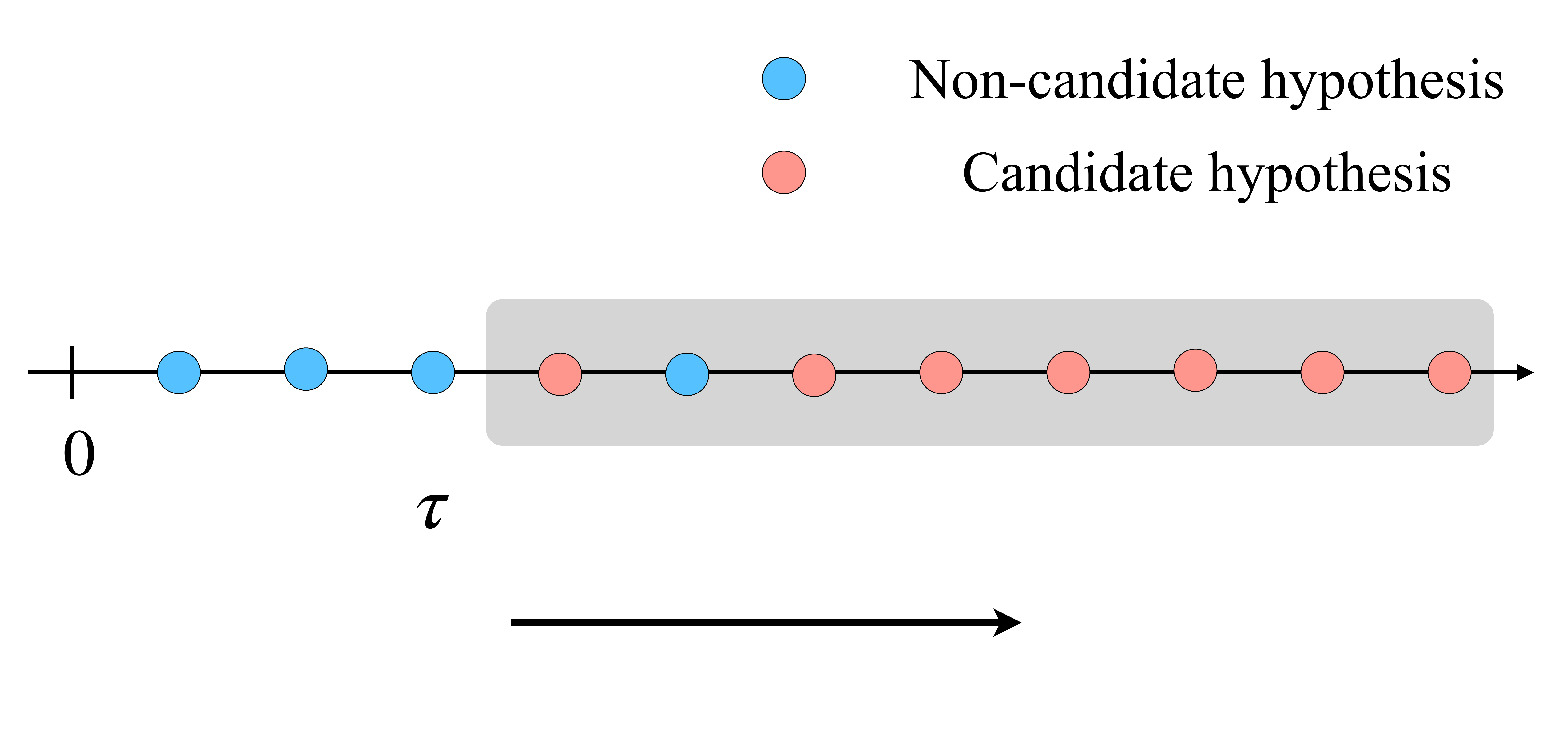}
    \caption{``Better'' ordering.}
  \end{subfigure}
\caption{Illustration of the step-up procedure with
                  different orderings. The nodes correspond to
                  hypotheses, and they are ordered from left to
                  right. The red nodes in Figure
                  \ref{fig:ordering_intro} represent the hypotheses
                  that are candidate for selection and the blue nodes
                  those that are not. (The notion of
                  ``candidate hypothesis'' is introduced later in
                  Section \ref{subsec:modelx}.)
                  The procedure operates sequentially, stops when the
                  ratio of the number of blue nodes to the number of
                  red nodes falls below a threshold $q$, and selects
                  the remaing red nodes.  With $q=0.2$, the ordering
                  on the left yields five discoveries while that on
                  the right yields seven (the stopping point is marked
                  by $\tau$).}
\label{fig:ordering_intro}
\end{figure}

While we shall explore how to design orderings that exploit side information
in Section \ref{subsec:description}, we first demonstrate how this can
be applied to a genome-wide association study (GWAS). We consider the
dataset provided by the Wellcome Trust Case Control Consortium
\citep{wellcome2007genome}, which contains genetic information on
$n = 4913$ British individuals, of which $1917$ have Crohn's disease
and $2996$ are healthy controls. For each individual, $p = 377,749$
single nucleotide polymorphisms (SNPs) are recorded. Our inferential
goal is to discover SNPs that are significantly associated with
Crohn's disease in the British population (i.e. to discover non-nulls)
by means of a procedure controlling the FDR below the threshold
$q=0.1$.

The WTCCC dataset has been studied in several works, see e.g.,
\citet{wellcome2007genome, candes2018panning, sesia2018gene}, with the
last two references using knockoff-based methods. We extend knockoffs
by leveraging summary statistics---p-values or z-scores corresponding
to marginal testing of each individual SNP---reported by genetic
studies of Crohn's disease in other populations. In this particular
example, we worked with summary statistics from GWAS in East Asia and
Belgium
\citep{franke2010genome,liu2015association,goyette2015high}.\footnote{The
  summary statistics are obtained from
  \url{https://www.ibdgenetics.org/downloads.html}.} {\em Since the summary statistics come from studies in other populations, note that we are not trying to re-discover SNPs that have been discovered before.}
	
The adaptive knockoff filter, or {\em adaptive knockoffs} for short,
uses both the WTCCC data and the summary statistics to order the
hypotheses. It then sequentially examines, stops and selects
hypotheses in pretty much the same way as we have seen before. Table
\ref{tab:discories} compares summary results on the WTCCC data, and we
can see that adaptive knockoffs discovers more SNPs than
other methods. Details including a full list of
discovered SNPs are available in Appendix \ref{appendix:snps}.

\begin{table}[h!]
  \centering
  \begin{tabular}{|c|c|}
    \hline
    Study/Method &  Number of SNPs discovered \footnotemark\\
    \hline
    \cite{wellcome2007genome} & 9 \\
    \hline
    \cite{candes2018panning} & 18\\
    \hline
    \cite{sesia2018gene} & 22.8 \\
    \hline
    \textbf{Adaptive knockoffs} & \textbf{33.3}\\
    \hline
  \end{tabular}
  \caption{Number of SNPs discovered to be associated with Crohn's
    disease by different methods. The target FDR level is $q=0.1$ in
    all cases (\citet{wellcome2007genome} considers the Bayesian
    FDR). Knockoff-based algorithms are randomized and, consequently,
    the reported numbers of discoveries are averaged over multiple
    realizations of the algorithm. In the case of adaptive knockoffs,
    the number of realizations is 50.}
  \label{tab:discories}
\end{table}

\paragraph{Inference is valid conditionally on the side information}
We wish to stress at the onset of this paper that adaptive knockoffs
controls finite-sample FDR regardless of the correctness of the side
information, i.e., regardless of the correctness of the summary
statistics in our example.  Even in the case where the side
information is plain wrong, we still achieve FDR control. When side
information is useful, power may be increased (as is the case above).
As we shall see, the reason is simple: FDR control and higher statistical
power both hold conditionally on the side information.

\section{Model-X knockoffs}\label{subsec:modelx}

Before presenting the details of adaptive knockoffs, we start by
giving a brief introduction to the model-X knockoffs framework. Assume
the covariates $X =(X_1,\ldots,X_p)$ follow a known joint distribution
$P_X$ and let $P_{Y|X}$ denote the conditional distribution of the
response $Y$ as before. The inferential goal is to test whether or not
$P_{Y|X}$ depends on $X_j$. It is shown in
\citet{edwards2012introduction} and \citet{candes2018panning} that under mild
conditions the above testing problem is equivalent to testing
\begin{align}\label{eqn:defnull}
  H_j: Y~\indep~ X_j|X_{-j},
\end{align}
where $X_{-j}\in\R^{p-1}$ is the vector $X$ after deleting
$X_j$. Hypothesis $j$ is called a null if $H_j$ is true and a non-null
otherwise. Hence, a variable is null if and only if it is independent
of the response given the knowledge of the others; throughout the
paper, we shall work with \eqref{eqn:defnull}.

The knockoffs procedure starts by computing a \emph{feature importance
  statistic} $W_j$ for each hypothesis $H_j$. Before constructing the
$W_j$'s, we first describe two key properties: (1) the null $W_j$'s
have equal probability of being positive or negative; (2) the signs of
the null $W_j$'s are mutually independent, and are independent of the
signs of the non-null $W_j$'s. Also, the feature importance statistics
are designed in such a way that the non-null $W_j$'s tend to take on
larger values. That said, we call $H_j$ a {\em non-candidate
  hypothesis} if $W_j<0$ and a {\em candidate hypothesis} if $W_j>0$
(as we have seen before, knockoffs only selects among the candidate
hypotheses).\footnote{The features with $W_j=0$ will never be selected
  or used by the procedure so we exclude them in the definitions.} The
vanilla knockoffs procedure then sorts the hypotheses by ordering the
magnitudes in a non-decreasing fashion,
$|W_{\pi_1}|\leq \ldots \leq |W_{\pi_k}|\leq \ldots|W_{\pi_p}|$, and
sequentially examines the hypotheses as follows: at each step
$k = 0, 1, 2, \ldots, p-1$, assume we select all remaining candidate
hypotheses $\pi_j$ for which $j > k$ and $W_{\pi_j} > 0$. Then the
number of false discoveries would be
$\#\{j:j>k,W_{\pi_j}>0,\pi_j\in\calH_0\}$. We do not have access to
this number since we do not know whether an hypothesis is null or not.
However, note that by symmetry of the null scores,
\[\#\{j:j>k,W_{\pi_j}>0,\pi_j\in \calH_0\}\approx
  \#\{j:j>k,W_{\pi_j}<0,\pi_j\in\calH_0 \}\leq
  \#\{j:j>k,W_{\pi_j}<0\}.
\]
Hence, the quantity
\begin{align}\label{efdp}
  \widehat{\fdr}_+(k) := \dfrac{1+\sum_{j > k} \indc{W_{\pi_j}<0}}{(\sum_{j > k}\indc{W_{\pi_j}>0})\vee 1}
\end{align}
may be regarded as a (conservative) estimate of the false discovery
proportion (FDP) among the unexamined
hypotheses. 
Set $[p] = \{1, \ldots, p\}$. Then the procedure is stopped at time
$T_+$, where
\begin{align}	
  T_+ : = \inf\{ k\in [p]: \widehat{\fdr}_+(k)\leq q\},
\end{align}
with the convention $\inf\varnothing= \infty$. The final selected set
is the family of remaining candidate hypotheses,
i.e.~$\hat{\calS} = \{\pi_j : j > T_+ , W_{\pi_j}>0\}$.
\citet{candes2018panning,barber2015controlling} established that this
procedure achieves FDR control at the nominal level
$q$. Alternatively, the quantity
\begin{align}\label{efdp0}
  \widehat{\fdr}_0(k): = \dfrac{\sum_{j > k} \indc{W_{\pi_j}<0}}{(\sum_{j >k}\indc{W_{\pi_j}>0})\vee 1},
\end{align}
is a slightly less conservative estimate of FDR. Replacing
$\widehat{\fdr}_+$ with $\widehat{\fdr}_0$ and replacing $T_+$ with
\begin{align}
  T_0 : = \inf\{ k\in[p] : \widehat{\fdr}_0(k)\leq q\},
\end{align}
yields control of a modified version of FDR defined as
\begin{align}
  \mathrm{mFDR}\, := \, \E \, \, \left[ \dfrac{|\hat{\calS}\cap \calH_0|}{|\hat{\calS}|+ q^{-1}}\right].
\end{align}
Figure \ref{fig:vkn_illustration} illustrates how the model-X knockoff
procedure orders, sequentially examines the hypotheses,
and stops when  $\widehat{\fdr}_0$ is below the pre-specified
threshold $q$.
		
\begin{figure}[ht]
  \begin{subfigure}{0.49\textwidth}
    \centering \includegraphics[width =
    1\textwidth]{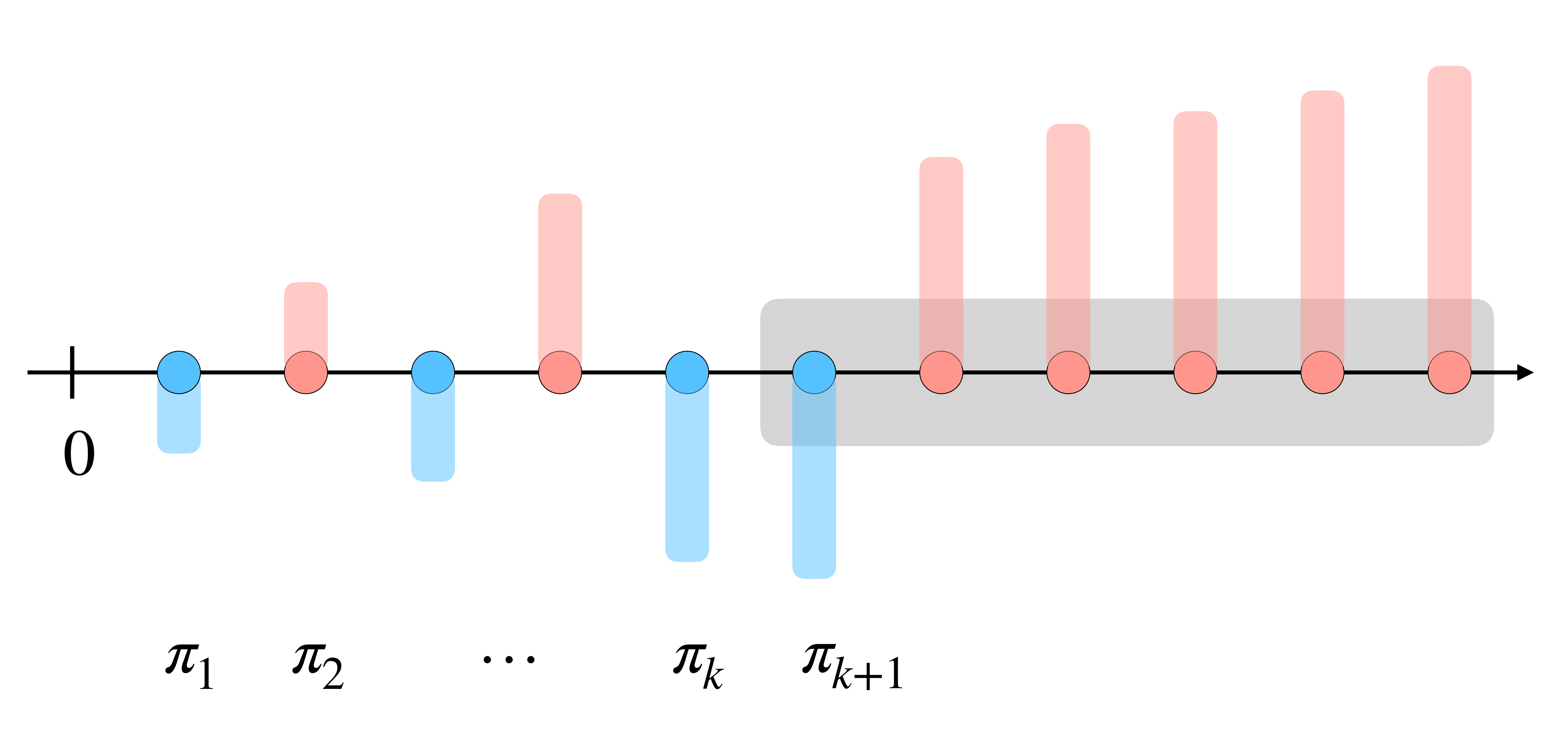}
    \caption{Model-X knockoffs.}
    \label{fig:vkn_illustration}
  \end{subfigure}
  \hfill
  \begin{subfigure}{0.49\textwidth}
    \centering \includegraphics[width =
    1\textwidth]{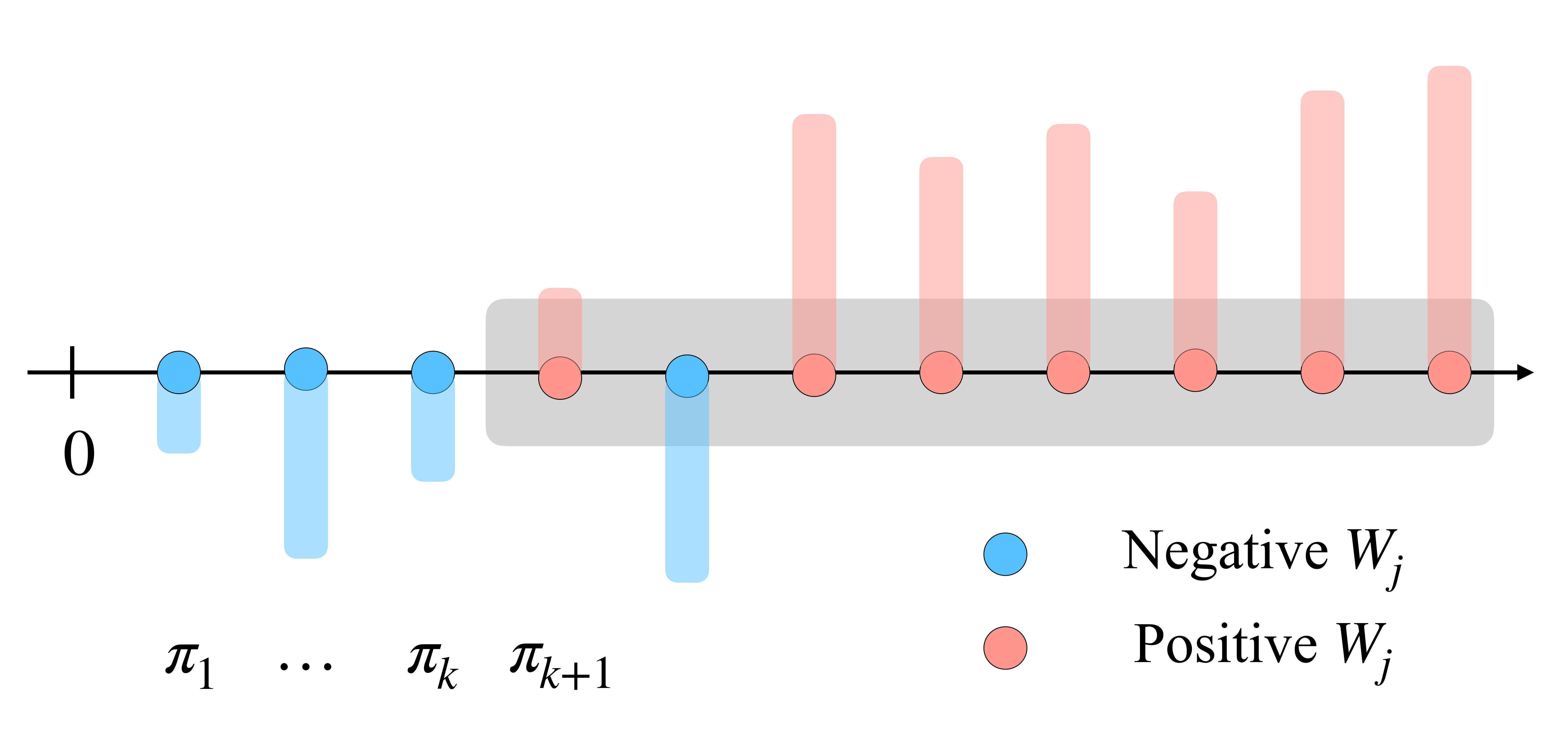}
    \caption{Adaptive knockoffs with side information.}
    \label{fig:akn_illustration}
  \end{subfigure}
  \caption{Illustration of knockoffs and adaptive knockoffs. The
    length of a bar represents the magnitude of a feature importance
    statistics $W_j$ whereas the color represents the sign. Red
    (resp.~blue) nodes and bars correspond to positive
    (resp.~negative) $W_j$'s. The target FDR level is $q=0.2$ and we
    use $\widehat{\fdr}_0$. (a) The ordering $\{\pi_k\}_{k\in [p]}$ is
    based on the magnitude of the feature importance statistics
    (standard procedure). The algorithm selects the last five
    hypotheses. (b) The ordering $\{\pi_k\}_{k\in [p]}$ is determined
    by both the side information and the magnitude of the feature
    importance statistics. The algorithm selects seven hypotheses.}
\end{figure}

We now briefly describe the computation of the feature importance
statistics.  Throughout the paper, assume we are given $n$
i.i.d.~samples from $P_{X} \cdot P_{Y|X}$.
For each sample $(X,Y)$, we augment the dataset by constructing a {\em
  knockoff copy}
$\tilde{X} = (\tilde{X}_1, \ldots, \tilde{X}_p) \in \R^p$ for
$X = (X_1, \ldots, X_p) \in \R^p$ (each original feature $X_j$ has a
knockoff copy $\tilde{X}_j$). This construction obeys two properties:
first, $\tilde{X} = (\tilde{X}_1, \ldots, \tilde{X}_p)$ is independent
of $Y$ conditional on $X$; second, the joint distribution
$(X,\tilde{X})$ remains invariant if we swap $X_j$ and $\tilde{X}_j$,
for any $j\in\calH_0$.  Formally,
$(X_j, \tilde{X}_j) | X_{-j}, \tilde{X}_{-j} \eqd (\tilde{X}_j, X_j) |
X_{-j}, \tilde{X}_{-j}$.
How to construct good knockoffs is an expanding area of research, see
e.g.~\cite{candes2018panning,sesia2018gene,gimenez2018knockoffs,liu2018auto,romano2019deep,SB-EC-LJ-WW:2019}. In
this paper, we will mainly be using the Gaussian
\citep{candes2018panning} and HMM knockoffs \citep{sesia2018gene}, and
point the reader to these references for details.

%

The knockoff variables should be thought of as some sort of negative
controls.  When the statistician wants to evaluate the effect of each
covariate on the response, she usually runs an algorithm on
$(\mathbf{X},\mathbf{Y})$---the covariate matrix and the response
vector---and obtains an importance score $Z_j$ for each feature
$j$. For example, $Z_j$ can be the magnitude of the Lasso coefficient
for $X_j$, with the value of the regularization parameter determined
by cross-validation. Now the knockoffs procedure asks our statistician
to run her algorithm on both the original {\em and} the knockoff
features.  She will now obtain two scores $Z_j$ and $\tZ_j$ for each
feature. In our previous example, the first is the magnitude of the
Lasso coefficient for $X_j$ and the second that for $\tX_j$ (we are
still free to determine the value of the regularization parameter by
cross-validation if we wish). She then combines these two scores into
a single one as follows:
\begin{align}
W_j = w_j(Z_j,\tZ_j);
\end{align}
here, $w_j$ is any anti-symmetric function she wants to use (e.g., $W_j = Z_j-\tZ_j$).\footnote{An
anti-symmetric function is a function such that
$f(u,v) = -f(v,u)$.} By construction, 
if $j$ is a null, $W_j$ has equal probability of being positive or
negative, whereas if $j$ is not null, we hope that $W_j$ tends to be
large and positive.
		
%

\section{The adaptive knockoff filter}\label{subsec:description}

As discussed before, the ordering $\{\pi_k\}_{k\in[p]}$ is a key
element in the knockoff procedure. If we know a priori that some
hypotheses are more likely to be non-nulls and move them towards the
end of the ordering, the procedure is more likely to select these
features. Now suppose side information associated with the features
under study is available. We would like to know
\begin{enumerate}[label = (\alph*)]
\item how we can effectively use the data and side information to
  construct an ordering that has higher density of non-nulls at the
  end (as to improve power), 
\item and what property should the ordering have so that the FDR
  remains controlled? 
\end{enumerate}
Informally, in order to keep FDR control, we require the ordering to
be independent of the signs of the statistics. Let $U_j\in \R^r$
denote the side information associated with feature $j$ and
$U = (U_1,\ldots,U_p)^T$. Let $V^+(k)$ (resp. $V^-(k)$) denote the
\emph{null} features with positive (resp. negative) test scores $W_j$ that have
not been examined up to and including step $k$. Define the filtration
$\{\calF_k\}_{k\geq 0}$, where $\calF_k$ is the $\sigma$-algebra
generated by the following elements:
\begin{itemize}
\item The magnitude of all the $W_j$'s: $\{|W_j|\}_{j\in[p]}$.
\item The signs of the examined $W_j$'s: $\{\sign(W_{\pi_j})\}_{j\leq k}$ (when $k = 0$, this is the empty set).
\item The signs of the non-null $W_j$'s :$\{\sign(W_j)\}_{j\in \calH_1}$.
\item The number of positive and negative null $W_j$'s in the
  unexamined hypotheses: $\{|V^+(j)|\}_{j\leq k}$ and
  $\{|V^-(j)|\}_{j\leq k}$.
\item Side information: $U$.
\end{itemize}

	

\begin{property}[Sign invariant property]\label{property:order}
An ordering $\{\pi_j\}_{j\in [p]}$ is called \textbf{sign invariant}
if for any $k\geq 0$, conditional on $\calF_{k}$ and $\pi_{k+1}\in
V^+(k)\cup V^-(k)$, the probability of $W_{\pi_{k+1}}>0$ is equal to
$|V^+(k)|/(|V^+(k)| + |V^-(k)|)$.

\end{property}

Algorithm \ref{algo.akn} presents the adaptive knockoffs procedure. At
each step $k=0, 1, 2, \ldots$, adaptive knockoffs uses a \emph{filter}
$\Phi_{k+1}$, required to be
$\calF_{k}$-measurable, 
to determine the least promising hypothesis among the
remaining ones.\footnote{When the filter $\Phi_k$ has extra
  randomness, we combine the extra randomness with the original side
  information and consider the augmented side information and the
  corresponding augmented $\sigma$-field $\tilde{\calF}_k$. By such
  treatment, $\Phi_k$ is measurable w.r.t. $\tilde{\calF}_{k-1}$.}
Under this condition, Proposition \ref{prop:filterproperty} shows that
the resulting ordering $\{\pi_k = \Phi_k\}_{k\in[p]}$ obeys Property
\ref{property:order}. Adaptive knockoffs otherwise adopts the same FDR
estimates as in \eqref{efdp} and \eqref{efdp0}, and is stopped the
first time the estimate falls below the target threshold.

		
\begin{proposition}\label{prop:filterproperty}
  Assume that conditional on the side information $U$, the null
  $W_j$'s have equal probability of being positive or negative, and
  that their signs are independent of each other and of those of the non-nulls. If for each $k\geq 0$, $\Phi_{k+1}$ is
  $\calF_{k}$-measurable, then the ordering $\pi_k = \Phi_k$ obeys
  Property \ref{property:order}. 
\end{proposition}
\begin{proof}
  By assumption, $\pi_{k+1}$ is measurable w.r.t.~$\calF_k$ and
  consequently $\{\pi_{k+1} \in V^+(k)\cup V^-(k)\}\subset
  \calF_k$. Apart from
  $(U,\{|W_j|\}_{j\in[p]},\{\sign(W_j)\}_{j\in\calH_1},\{\sign(W_{\pi_j})\}_{j\leq
    k})$, $\calF_k$ can only provide further information on the number
  of ``$+$''s and ``$-$''s in $V^+(k)\cup V^-(k)$;
  i.e.~$|V^\pm(k)|$. Since the signs of the nulls $W_j \neq 0$ are
  i.i.d.~coin flips conditional on
  $(U,\{|W_j|\}_{j\in[p]},\{\sign(W_j)\}_{j},\{\sign(W_{\pi_j})\}_{j\leq
    k})$, the probability of $W_{\pi_{k+1}}>0$ (resp.~$W_{\pi_{k+1}}<0$)
  is proportional to the number of ``$+$''s (resp.~``$-$''s), completing
  the proof. 
\end{proof}
We would like to remark that if each knockoff copy has the properties
that $(X,\tX)|U$ stays invariant after swapping $X_j$ and $\tX_j$ and
$\tX$ is independent of $Y$ conditional on $(X,U)$, then the $W_j$'s satisfy the conditions required in Proposition \ref{prop:filterproperty}.

		\begin{algorithm}[h!]
  \DontPrintSemicolon \SetAlgoLined \BlankLine
  \caption{Adaptive Knockoffs\label{algo.akn}}
  \textbf{Input:} Covariate matrix $\mathbf{X}\in \R^{n\times p}$; response
  variables $\mathbf{Y}\in \R^{n}$; side information $U\in \R^{p\times r}$;
  target FDR level $q$. \; \textbf{Initialization:} $k\leftarrow 0$;
  $\widehat{\fdr}$ is either $\widehat{\fdr}_0$ or
  $\widehat{\fdr}_+$.\; \While{$\widehat{\fdr}(k)>q$ \normalfont{and}
    $k<p$}{ 1. Use the filter $\Phi_{k+1}$ to determine the next
    hypothesis to examine $\pi_{k+1}$:
    \begin{align}
      \pi_{k+1} \leftarrow \Phi_{k+1}(\{|W_{j}|\}_{j\in[p]},\{W_{\pi_j}\}_{j>k},U).			
    \end{align}\;
    2. Update $k$: $k\leftarrow k+1.$\; }
  \textbf{Output}: Selected set $\hat{\calS}=\{j\in[p]:j>k, W_{\pi_j}>0 \}$.
\end{algorithm}
		
As a result of Proposition \ref{prop:filterproperty}, we show in
Theorem \ref{thm:fdr} that adaptive knockoffs controls the
finite-sample FDR.
\begin{theorem}\label{thm:fdr}
Under the conditions from Proposition \ref{prop:filterproperty},
when $\widehat{\fdr}_+$ is used, Algorithm \ref{algo.akn} controls
the FDR at the nominal level $q$; when $\widehat{\fdr}_0$ is used, it controls the
modified FDR at  level $q$.
\end{theorem}
		
This result is a generalization of the condition for FDR control in
the knockoffs framework presented in \citet{barber2015controlling} and
\citet{candes2018panning}. The vanilla knockoff filter, which only uses
the magnitude of feature importance statistics to determine the order,
can be viewed as
a special case of adaptive knockoffs: in this case, 
\begin{align}
  \Phi_{k+1} = \argmin{j>k }~|W_{\pi_j}|,
\end{align}
which is clearly $\calF_{k}$-measurable.
			
\paragraph{Proof of Theorem \ref{thm:fdr}} When $\widehat{\fdr}_+$ is
used,
\begin{align}
\fdr = \E \left[\dfrac{|\hat{\calS}\cap \calH_0|}{|\hat{\calS}|}\right] = \E \left[\dfrac{|V^+(T_+)|}{|\hat{S}|}\right]&\leq  \E \left[\dfrac{|V^+(T_+)|}{|V^-(T_+)|+1}\widehat{\fdr}_+(T_+)\right]\\
& \leq q\E\left[\dfrac{|V^-(T_+)|}{|V^+(T_+)|+1}\right]\\
&\leq q.
\end{align}
The second inequality holds by  definition of $T_+$ and the third
inequality follows from the fact that $\dfrac{|V^-(k)|}{|V^+(k)|+1}$ is a
supermartingale and $T_+$  a stopping time w.r.t.~the filtration
$\{\calF_k\}_{k\geq 0}$. The supermartingale argument follows directly from Proposition
\ref{prop:filterproperty} and \citet[Section A.1]{barber2015controlling}. The
proof of mFDR control is exactly the same as in
\citet[Section A.2]{barber2015controlling}. 
			


\section{Two classes of filters}\label{subsec:filters}
We now focus on constructing a filter that satisfies Property
\ref{property:order} and also systematically uses all the available
information to determine the ordering of hypotheses. At each step $k$,
the filter determines the ``least promising'' hypothesis among the
unexamined hypotheses based on the information in $\calF_k$. We
present two types of filters that quantify ``least promising'' in
different ways. We emphasize that the model we choose does not affect
the FDR control as long as Property \ref{prop:filterproperty} is
satisfied, and researchers are free to come up with other types of
models. In the following, we refer to this situation with the slogan:
``Wrong models do not hurt FDR control!'' We also assume we work with
standardized side information $U_j$'s, which means that the $U_j$'s
have the same dimension and units.

\subsection{Predictive modeling}

At step $k$, we estimate the probability that the sign of a feature
importance statistic is negative conditional on
$\calF_k$. Specifically, we let $s_j = \sign(W_j)$ and compute an
estimate of $\p(s_j = -1|\calF_k)$ for each remaining feature. This
estimation (or prediction) task can be handled by various machine
learning algorithms. We treat $\{s_j\}^p_{j=1}$ as the binary
responses and the magnitude $\{|W_j|\}^p_{j=1}$ and side information
$\{U_j\}_{j=1}^p$ (e.g., $U_j$ is the prior rank of $H_j$) as
predictors. We consider the model
\begin{align}
  g(\p(s_j = 1||W_j|,U_j)) 	= h(|W_j|,U_j),
\end{align}
where $g(x) = \log(x/(1-x))$ is the link function\footnote{In the case
  where $W_j$ can also be $0$, we can alternatively use a multinomial
  model with levels $\{-1,0,1\}$.} and $h(\cdot, \cdot)$ is a
regression function. If we postulate a logistic model, 
\begin{align}
  h(|W_j|,U_j) =  \beta_0+\beta_1|W_j|+\beta_2^TU_j.
\end{align}
For a generalized additive model (GAM),
\citep{hastie2017generalized},
\begin{align}
  h(|W_j|,U_j) = \beta_0+ h_0(|W_j|)+h_1(U_{j1})+\ldots+h_r(U_{jr}),
\end{align} 
where $h_0,h_1,\ldots,h_r$ are smooth functions from $\mathbb{R}$ to
$\mathbb{R}$.  The function $h$ can also be modeled via random forests
\citep{breiman2001random}.

We use
$(\{s_{\pi_j}\}_{j\leq k}, \{U_j\}_{j\in [p]},\{|W_j|\}_{j\in [p]})$
as training data to fit the chosen model, and the fitted function
$\hat{h}$ for predicting the signs of statistics among the unexamined
hypotheses. For $j> k$, set
\begin{align}
  \hat{\p}(s_{\pi_j} = -1||W_{\pi_j}|,U_{\pi_j}) = g^{-1}\circ \hat{h}(|W_{\pi_j}|,U_{\pi_j}),
\end{align}
and 
\begin{align}
  \Phi_{k+1} = \argmax{j>k}~ g^{-1} \circ
  \hat{h}(|W_{\pi_j}|,U_{\pi_j}) =  \argmax{j>k}~ 
  \hat{h}(|W_{\pi_j}|,U_{\pi_j})
\end{align}
since $g$ is monotone. 
By construction, $\Phi_{k+1}$ is $\calF_{k}$-measurable.
			
\subsection{Bayesian modeling}\label{sec:bayesfilter}
An alternative perspective, which has the benefit of allowing for a
careful modeling of the effect of side information, is of a Bayesian
nature. That said, we are not imposing any assumption on the data
generating mechanism. We are simply using Bayesian thinking for
calculating the probability of a feature being non-null, and whether
the Bayesian beliefs about features are true or not does not hurt
FDR control. A belief closer to the truth will yield higher power in
detecting the non-nulls.

\paragraph{The model} 
The Bayesian-oriented filter is similar to the treatment in
\citet{lei2018adapt}, but we consider it for knockoffs. Let $H_j$
denote whether or not feature $j$ is a null: $H_j=1$ means feature $j$
is a non-null and $H_j=0$ means it is a null. We follow the Bayesian
two-group model and write
\begin{align}
  H_j|U_j \overset{i.i.d.}{\sim} \Bern(\nu(U_j)),
\end{align}
where $\nu$ is a link function. Marginally,
\begin{align}
  W_j|H_j,U_j \sim \begin{cases}
    \calP_1(W_j|U_j) & \text{if }H_j = 1,\\
    \calP_0(W_j|U_j) & \text{if }H_j = 0.
  \end{cases}
\end{align}
Above, $\calP_{H_j}(\cdot|U_j)$ denotes the law of $W_j$ conditional
on $U_j$ when $H_j\in\{0,1\}$. Under this model, we can quantify the
possibility of a feature being null by inspecting the posterior
probability $\p(H_j=0||W_j|,U_j)$.
At each step $k$, the posterior probability can be used as a criterion
to determine the next hypothesis in the ordering, i.e.,\footnote{In
  implementation, we instead use
  $1- \p(H_j=1,\sign(W_j)>0||W_j|,U_j)$.}
\begin{align}
  \Phi_{k+1} = \argmax{j>k}\quad  \p(H_{\pi_j}=0||W_{\pi_j}|,U_{\pi_j}).
\end{align}
The remaining task is to model $\calP_0,\calP_1$ and $\nu$. 
Assuming $W_j$
has a distribution with a point mass at $0$, we model the conditional
law of $W_j$ via
\begin{align}
  p_h(w|u) = \delta_h\indc{w=0}+(1-\delta_h)\indc{w\neq 0}\dfrac{\beta_h(u) \exp(\beta_h(u)w)}{(1+\exp(w))^{\beta_h(u)+1}},\qquad h = 0,1.
\end{align}
The continuous part of the distribution is somewhat arbitrary, and we
choose this form for computational convenience. Under this model, 
\begin{align}\label{eq:condexp}
& \E[H|U] = \nu(U),\\
& \E[Y|W\neq 0,U,H=h] =1/\beta_h(U),\qquad h=0,1,
\end{align}
where $Y_j = \log(1+\exp(W_j))-W_j$. Then estimating
$(\nu(U),\beta_0(U),\beta_1(U))$ boils down to estimating the above
conditional expectations. 


\paragraph{GLM-based approach} 
Let $\calN$ (resp.~$\calB$) denote the class of functions $\nu(\cdot)$
(resp.~$\beta_0(\cdot)$, $\beta_1(\cdot)$) belongs to.  For example,
assuming a logistic model, we have
\begin{align}\label{eq:modelnu}
\calN = \{\nu(x): \nu(x)= 1/(1+\exp(-\theta^T x)),~ \theta\in\R^d \}
\end{align}
while a model for $\calB$ might be
\begin{align}\label{eq:modelbeta}
\calB = \{\beta(x):\beta(x) = \exp(\theta^Tx),~\theta\in\R^d\}.
\end{align}
The log-likelihood function (under independence) of
$\{(H_j,W_j)\}_{j\in [p]}$ conditional on $\{U_j\}_{j\in [p]}$ is
given by {\small
\begin{align}
\ell(\{H_j,W_j\}_{j\in [p]}|\{U_j\}_{j\in [p]};\delta_0,\delta_1,\nu(\cdot),\beta_0(\cdot),\beta_1(\cdot)) = & \sum^p_{j=1}[ \underbrace{(i) + (ii)}_{\text{group 1}}+\underbrace{(iii)+(iv)+(v)}_{\text{group 2}}]+C,
\end{align}	
}where $C$ represents the terms not containing the parameters and
group $1$ includes
\begin{align}
\begin{cases}
&(i) =(1-H_j)\indc{W_j =0}\log(\delta_0) +(1-H_j)\indc{W_j\neq 0}\log(1-\delta_0),\\
&(ii)  = H_j\indc{W_j =0}\log(\delta_1) +H_j\indc{W_j\neq 0}\log(1-\delta_1).			 
\end{cases}
\end{align}
Group $2$ comprises
\begin{align}
\begin{cases}
&(iii)  = H_j\log(\nu(U_j)) + (1-H_j) \log(1-\nu(U_j)),\\
&(iv) =  (1-H_j)\indc{W_j\neq 0}(\log(\beta_0( U_j))+\beta_0(U_j)\log(\exp(W_j)/(1+\exp(W_j))),\\ 
&(v)  = H_j\indc{W_j\neq 0}(\log(\beta_1(U_j))+ \beta_1(U_j)\log(\exp(W_j)/(1+\exp(W_j))).
\end{cases}
\end{align}
In the case where $\nu(\cdot),\beta_0(\cdot),\beta_1(\cdot)$ are
classes of parametric functions as above, we hope to obtain the
maximum likelihood estimator (MLE) by optimizing the log-likelihood
function:
\begin{align}
(\hat{\delta}_0,\hat{\delta}_1,\hat{\nu}(\cdot),\hat{\beta}_0(\cdot),\hat{\beta}_1(\cdot)) = \argmax{\substack{\delta_0,\delta_1,\\\nu(\cdot)\in\calN,\\\beta_0(\cdot),\beta_1(\cdot)\in\calB}} ~\sum^p_{j=1}[(i)+(ii)+(iii)+(iv)+(v)].
\end{align}
Note that at step $k$, the information we can use to
estimate the parameters is limited: some of the signs
of the $W_j$'s are not available and the $H_j$'s are
unobserved.

Directly optimizing the log-likelihood function is not
feasible. Instead, we use the expectation-maximization (EM) algorithm
to obtain the MLE. Our plan is this: at step $k$ of the adaptive
knockoffs algorithm, we run the EM algorithm for $S$ iterations and
obtain an estimate of the parameters of interest. (To be clear, one
iteration of the EM algorithm consists of an E-step and an M-step.) At
step $s$ of the EM algorithm, denote by $\calG$ the $\sigma$-field
generated by the available information. For the E-step we
need to compute the following conditional expectations:
\begin{align}
\E[H_j|\calG],~\E[Y_jH_j|\calG],~\E[Y_j(1-H_j)|\calG].
\end{align}
We defer the calculation of the above quantities to Appendix
\ref{appendix:supp} and set $\overbar{H}_j = \E[H_j|\calG]$. For the
M-step, we decompose the optimization into two subgroups. The
optimization problems in group 1 have analytical solutions, namely, 
\begin{align}\label{opti:analytical}
&\hat{\delta}_0  = \argmax{\delta_0}~\sum^p_{j=1}(i) =  \dfrac{\sum_{j\in \calH}(1-\overbar{H}_j)\indc{W_j = 0} }{\sum_{j\in \calH} (1-\overbar{H}_j)},\\
&\hat{\delta}_1 = \argmax{\delta_1}~\sum^p_{j=1}(ii) =  \dfrac{\sum_{j\in \calH}\overbar{H}_j\indc{W_j = 0} }{\sum_{j\in \calH} \overbar{H}_j}.
\end{align}
The optimization problems in group 2 update 
$(\nu(\cdot),\beta_0(\cdot),\beta_1(\cdot))$. Since the optimization problem is
separable, we can solve the three subproblems independently. 
\begin{align}\label{opti:not}
	\hat{\nu}(\cdot) = \argmax{\nu(\cdot)\in\calN}~\sum^p_{j=1}(iii), ~~\hat{\beta}_{0}(\cdot) = \argmax{\beta_0(\cdot)\in\calB}~\sum^p_{j=1}(iv), ~~\hat{\beta}_{1}(\cdot) = \argmax{\beta_1(\cdot)\in\calB}~\sum^p_{j=1}(v).
\end{align}
These three subproblems directly depends on $\calN$ and $\calB$. When
the parametric model as in \eqref{eq:modelnu} and \eqref{eq:modelbeta}
is used, the above optimization problems correspond to three
(weighted) GLMs respectively and can be solved by standard R packages
(e.g., \textsf{glm}). 

\paragraph{GLM-extension approach}
Another possibility is to work with regularized log-likelihood
  functions. For instance, we may add an $\ell_1$ penalty about the
  coefficients $\theta$ in \eqref{eq:modelbeta}, and use the
  \textsf{glmnet} package to solve the corresponding optimization
  problem. We can also fit a generalized additive model by for $\beta_0(\cdot)$ solving the
  following penalized optimization problem \citep[Chapter
  9]{hastie2009elements}:
  \begin{align}
\max_{\beta_0(\cdot)\in\calB}~~ \sum^p_j (iv)-\sum_{\ell=1}^r\lambda_\ell\int \beta_{0,\ell}''(x_\ell)^2\mathrm{d}x_\ell,
\end{align}
in which
$\calB = \{\beta(x_1, \ldots, x_r) : \beta(x_1, \ldots, x_r) =
\sum^r_{\ell=1}\beta_\ell(x_\ell) ,~\beta_\ell''(\cdot) \text{ exists
  for all }\ell\in[r]\}$.
Above, the nonnegative hyper-parameters $\{\lambda_\ell\}_{\ell\in[r]}$ can be
chosen via Generalized Cross Validation (GCV). The R package
\textsf{gam} or \textsf{mgcv} are designed to find solutions to such
problems.

\paragraph{Nonparametric regression approach}
We consider a variation that does not fall in the EM framework but
allows us to make use of flexible regression tools.  Recall
\eqref{eq:condexp}, which states that
$(\nu(\cdot),\beta_0(\cdot),\beta_1(\cdot))$ are functions of the
conditional expectations. We thus directly estimate the conditional
expectation instead of solving the optimization problems in group
2. For example, we can use non-parametric methods, e.g., a random
forest, to directly fit the conditional expectations and let the
fitted values be the updated parameters. This is not an M-step because
we are no longer optimizing the (expected) likelihood. (This is not a
concern since FDR control always holds.) Such a variation opens the
door to modern regression methods and often works well in practice as
we shall see later.

\paragraph{Default implementation} 
The methods we have presented differ in the way they estimate
$\nu(\cdot)$, $\beta_0(\cdot)$ and $\beta_1(\cdot))$. When the side
information is a scalar, the default implementation combines the
GLM-extension and nonparamatric regression approaches. In details, we
fit $\beta_0(\cdot),\beta_1(\cdot)$ via the \textsf{gam} package in R
whereas for $\nu$, we regress $\log(\overbar{H}_j/(1-\overbar{H}_j))$ on $U_j$
via a GAM and then transform the fit to produce
$\hat{\nu}(\cdot)$. When the dimension is higher, the default
implementation is the nonparametric regression approach with a random
forest. The default number of iterations $S$ is set to be one.

\paragraph{Initialization}\label{par:initialization} At the
beginning of Algorithm \ref{algo.akn}, we reveal a fraction (by
default $10\%$) of the hypotheses based \emph{only} on the magnitude
of the statistics $|W_j|$ corresponding to the lowest values. Denote
the revealed statistics by $W_{\reveal}$. The adaptive filter then
starts with rough guesses
$(\hat{\beta}_0(\cdot),\hat{\beta}_1(\cdot),\hat{\nu}(\cdot),\hat{\delta}_0,\hat{\delta}_1)$
computed from available information. Specifically, we initialize
$\hat{\nu}(\cdot)$ with a constant function set to
${|\{j:W_j>0\}|}/{p}$ (we can think of a $|\{j:W_j>0\}|$ as a very
liberal estimate of the number of non-nulls).  Further, we set
\[
\hat{\delta}_0 = \frac{|\{j:W_j\leq 0\}|}{p} \frac{|\{j:W_j=0\}|}{p},
\quad 
\hat{\delta}_1 = \frac{|\{j:W_j> 0\}|}{p} \frac{|\{j:W_j=0\}|}{p}.
\]
Finally, the initial values of $(\hat{\beta}_0,\hat{\beta}_1)$ are
given by
\begin{align}
\hat{\beta}_{0}(U_j) &= \hat{\beta}_{1}(U_j) = 1/\log(2),~~\text{if }W_j=0.\\
\hat{\beta}_{0}(U_j) &= 1/[\log(1+\exp(\overbar{W}_\reveal^-))-\overbar{W}_\reveal^-],\\
 \hat{\beta}_{1}(U_j) &= 1/[\log(1+\exp(\overbar{W}_\reveal^+))-\overbar{W}_\reveal^+],~~\text{if }W_j\neq 0.
\end{align}
Above, $\overbar{W}_\reveal^-$ is the average of the negative items in
$W_{\reveal}$ and $\overbar{W}_\reveal^+$ is the average of the
positive items in $W_{\reveal}$. That is, we approximate
$\beta_0(U_j)$ (resp.~$\beta_1(U_j)$) with $1/Y_j$, in which we impute
nonzero $W_j$'s with the average of the negative (resp.~positive)
items in $W_{\reveal}$. 

In the subsequent steps of the filter, the initial value of the tuple
$(\hat{\beta}_0,\hat{\beta}_1,\hat{\nu},\hat{\delta}_0,\hat{\delta}_1)$
in Algorithm \ref{algo.em} is the output from the previous
iteration. The complete procedure is described in Algorithm
\ref{algo.em}.

\begin{algorithm}[h!]
	\DontPrintSemicolon  
	\SetAlgoLined
	\BlankLine
	\caption{EM algorithm\label{algo.em} to estimate
          $p_0,p_1,\nu$}
	\textbf{Input:}
        Information $\calF_{k}$ at step $k$. 
      
        \textbf{Initialization:} initialize ($\hat{\beta}_0$,
        $\hat{\beta}_1$, $\hat{\nu}$, $\hat{\delta}_0$
        $\hat{\delta}_1$) as in Section \ref{par:initialization} and
        set
        $\calG \leftarrow
        \sigma(\calF_{k},\hat{\beta}_0,\hat{\beta}_1,\hat{\nu},\hat{\delta}_0,\hat{\delta}_1)$.
        
        \For{$s\leftarrow 0,\ldots,S-1$}{
          1. \textbf{E-step: }\\
          ~~Update $\overbar{H}_j: ~ \overbar{H}_j \leftarrow \E[H_j|\calG],\hfill j\in[p].$\\
          ~~Update $(\overbar{Y}_{0,j},\overbar{Y}_{1,j}): ~ \overbar{Y}_{h,j}\leftarrow \dfrac{\E[Y_j(hH_j+(1-h)(1-H_j))|\calG]}{h\overbar{H}_j+(1-h)(1-\overbar{H}_j)},\hfill h=0,1,$\\
          where the calculations of the conditional expectations are
          presented in Appendix \ref{algo.em}.\;
          2. \textbf{M-step:}\\
          ~~Update $(\hat{\delta}_0,\hat{\delta}_1): ~  \hat{\delta}_h \leftarrow \dfrac{\sum_{j\in \calH}(h\overbar{H}_j+(1-h)(1-\overbar{H}_j))\indc{W_j = 0} }{\sum_{j\in \calH} (h\overbar{H}_j+(1-h)(1-\overbar{H}_j))},\hfill h=0,1.$\\
          ~~Update
          $\hat{\nu}: ~ \hat{\nu} \leftarrow \text{random
            forest}(\overbar{H}_j\sim
          U_j)$.\\
          ~~Update $(\hat{\beta}_0,\hat{\beta}_1): ~ 1/\hat{\beta}_{h} \leftarrow \text{random forest}(\overbar{Y}_{h,j}|{W_j\neq 0}\sim U_j),\hfill h=0,1$.\\
          3. \textbf{Update current information:}
          $\calG\leftarrow
          \sigma(\calF_k,\hat{\beta}_0,\hat{\beta}_1,\hat{\nu},\hat{\delta}_0,\hat{\delta}_1)$.
        } \textbf{Output:
          $\hat{\delta}_0,\hat{\delta}_1,\hat{\nu},\hat{\beta}_0,\hat{\beta}_1$}.
\end{algorithm}


      
\section{Numerical results}
\subsection{General setting}\label{sec:generalsetting}

To evaluate the performance of adaptive knockoffs, we present two
numerical experiments with different types of side information. In
each setting, we compare adaptive knockoffs with other multiple
testing methods. Table \ref{table:candidates} lists all the candidate
methods and their properties, i.e., whether or not they depend on
p-values and whether or not they utilize side information. In our
experiments all the p-values are obtained from multivariate linear
regression.  Storey-BH is implemented with a threshold set to
$\tau = 0.5$.  The parameter of SABHA follows \citet{li2019multiple}
with $\epsilon = 0.1$ and $\tau = 0.5$. For Adaptive SeqStep, the
threshold $\lambda$ is set to be $0.5$ as in \citet{lei2016power}.\footnote{The code for implementing BH, Storey-BH, SABHA and Adaptive SeqStep is adapted from \url{https://www.stat.uchicago.edu/~rina/sabha/All_q_est_functions.R} and \url{https://github.com/lihualei71/adaptPaper/blob/master/R/other_methods.R}.}
For AdaPT, we follow the setup introduced in
\url{https://cran.r-project.org/web/packages/adaptMT/vignettes/adapt_demo.html}.The
knockoff-based algorithms in Table \ref{table:candidates} use the LCD
feature importance statistics as introduced in
\citet{candes2018panning} and $\widehat{\fdr}_+$ as the estimated FDR.

For both experiments, we run algorithms with target FDR levels $\{0.03,0.06,\ldots,0.3\}$ and compare the corresponding statistical power and realized FDR. All the presented results are averaged over $100$ trials. The simulation results can be reproduced with the code provided at \url{https://github.com/zhimeir/adaptive_knockoff_paper}.
\begin{table}[h!]
\centering
\begin{tabular}{|c|c|c|c|}
\hline
Method & Abbreviation & \makecell{P-value\\free?} & \makecell{Use side\\ information?}\\
\hline
Benjamini Hochberg & BHq &  & \\
\hline
Storey's BH & StoreyBH &  & \\
\hline
Adaptive SeqStep & AdaSeqStep &  & $\checkmark$\\
\hline
AdaPT & AdaPT & & $\checkmark$ \\
\hline
Structure Adaptive BH algorithm & SABHA &  & $\checkmark$\\
\hline
Vanilla Model-X knockoffs & Vanilla Knockoff &$\checkmark$ &\\
\hline
Adaptive knockoffs w/ GLM filter & AdaKn(GLM) &$\checkmark$ &$\checkmark$ \\
\hline
Adaptive knockoffs w/ GAM filter & AdaKn(GAM) &$\checkmark$ &$\checkmark$ \\
\hline
Adaptive knockoffs w/ Random Forest filter & AdaKn(RF)& $\checkmark$&$\checkmark$ \\
\hline
Adaptive knockoffs w/ two group model & AdaKn(EM)& $\checkmark$&$\checkmark$ \\
\hline
\end{tabular}
\caption{Candidate multiple testing methods and their properties.}
\label{table:candidates}
\end{table}

\subsection{Simulation 1: one-dimensional side information}\label{sec:simul1}
The simulated dataset is of size $n = 1000$ and
$p=900$. Conditional on $X$, $Y$ is generated from a linear model 
\begin{align}\label{eqn:linearmodel}
Y | X_1, \ldots, X_p \sim \mathcal{N}(\beta_1 X_1 + \ldots \beta_p
  X_p, 1). 
\end{align}
The covariates $X$ are drawn from an HMM, whose parameters follow the
instructions found at
\url{https://msesia.github.io/snpknock/articles/SNPknock.html}. Researchers
can reproduce our choices by following the link from Section
\ref{sec:generalsetting}.  In this setting, our inferential goal is to
test whether or not $\beta_j = 0$.

We specify the model by constructing a sparse regression sequence
$\beta$---fixed throughout, i.e. through the $100$ trials so that the
data distribution $P_{XY}$ does not change---as follows: we randomly
choose 150 features among the first 300 as signals in such a way that
the larger the index, the less likely it is to be selected.\footnote{ 
We draw i.i.d.~samples from a distribution supported on
$\{1,2,\ldots,300\}$ such that $j$ is selected with probability
proportional to $\frac{1}{j^2}$ until we obtain $150$
distinct realizations. }
The setting is motivated by the fact
that in many real applications, researchers have access to prior
knowledge about the hypotheses, which allows them to rank the hypotheses by
their chance of being of interest. For each signal $X_j$, we set
$\beta_j = \pm 3.5/\sqrt{n}$, where the signs are determined by
independent coin flips (the features not in the model have
$\beta_j = 0$).  Figure \ref{fig:sim1_str} shows the realized
configuration of the signals (the variables with nonzero regression
coefficients). The side information is the index of the features; that
is, $U_j = j$ for $j\in[p]$.

\begin{figure}[h!]
\centering
\begin{subfigure}{0.8\textwidth}
\centering
\includegraphics[width = 0.6\textwidth]{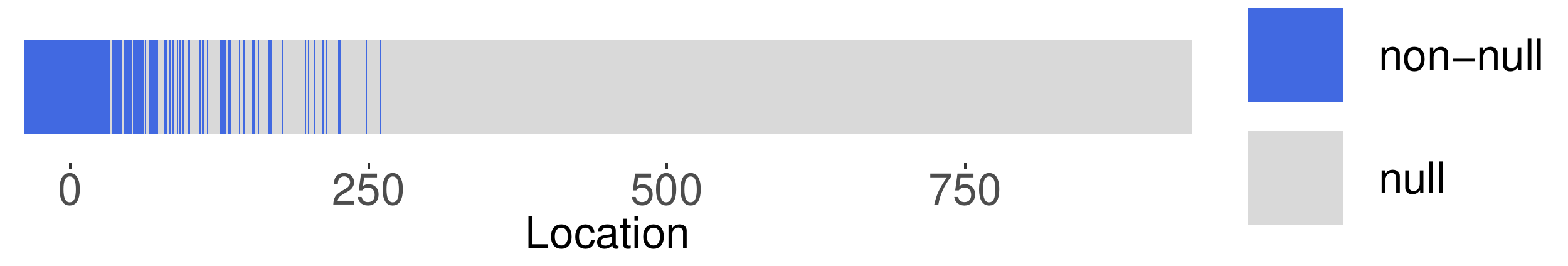}
\caption{Hypothesis structure.}
\label{fig:sim1_str}
\end{subfigure}
\\	\begin{subfigure}{0.8\textwidth}
\centering
\includegraphics[width = 0.6\textwidth]{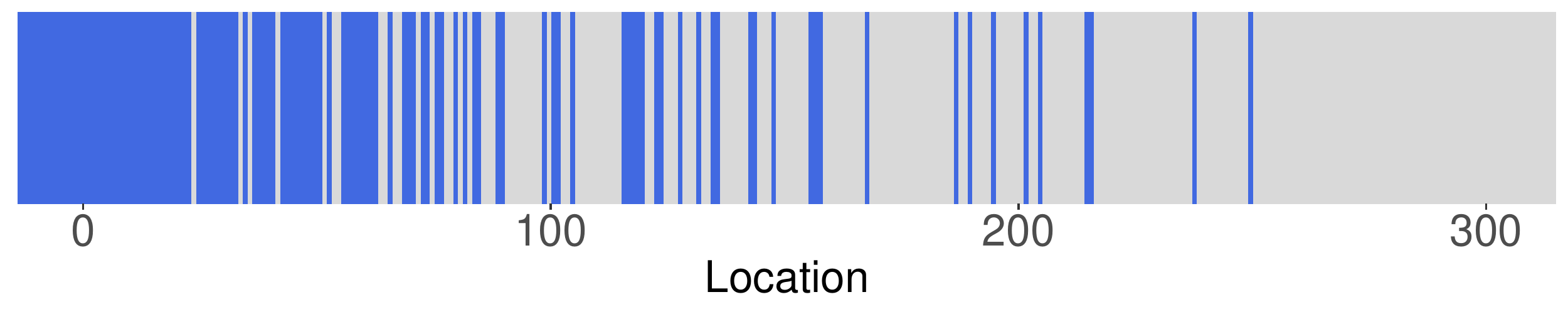}
\caption{Zoomed-in view of the first $300$ indices.}
\label{fig:sim1_str_zoomin}
\end{subfigure}
\caption{One-dimensional hypothesis structure.}
\end{figure}

In each trial, we draw a sample of size $n = 1000$ from $P_{XY}$ and
run all candidate methods on this sample. Figure \ref{sim1_res} shows
the power and FDR of each method versus target FDR levels. All methods
control FDR as we expected. The adaptive knockoffs outperforms vanilla
knockoffs and other p-value based procedures by a wide margin. We also
plot the realized ordering of vanilla knockoffs and adaptive knockoffs
(with our Bayesian filter) in Figure \ref{fig:vordering1} and Figure
\ref{fig:adaordering1} respectively. We can observe that adaptive
knockoffs places more non-nulls towards the end of the ordering and,
consequently, makes more true discoveries.

The p-value based methods perform unsatisfactorily here because
p-values are of low quality. As an aside, we note that it is often
challenging to obtain valid p-values, not to mention high quality
ones; for instance, \citet{dezeure2015high} and
\citet{lei2019assumption} explain that getting p-values from the
simplest linear model in reasonably high dimensions is already a
challenge if we do not impose stringent assumptions.

\begin{figure}[h!]
\centering
\includegraphics[width = 1\textwidth]{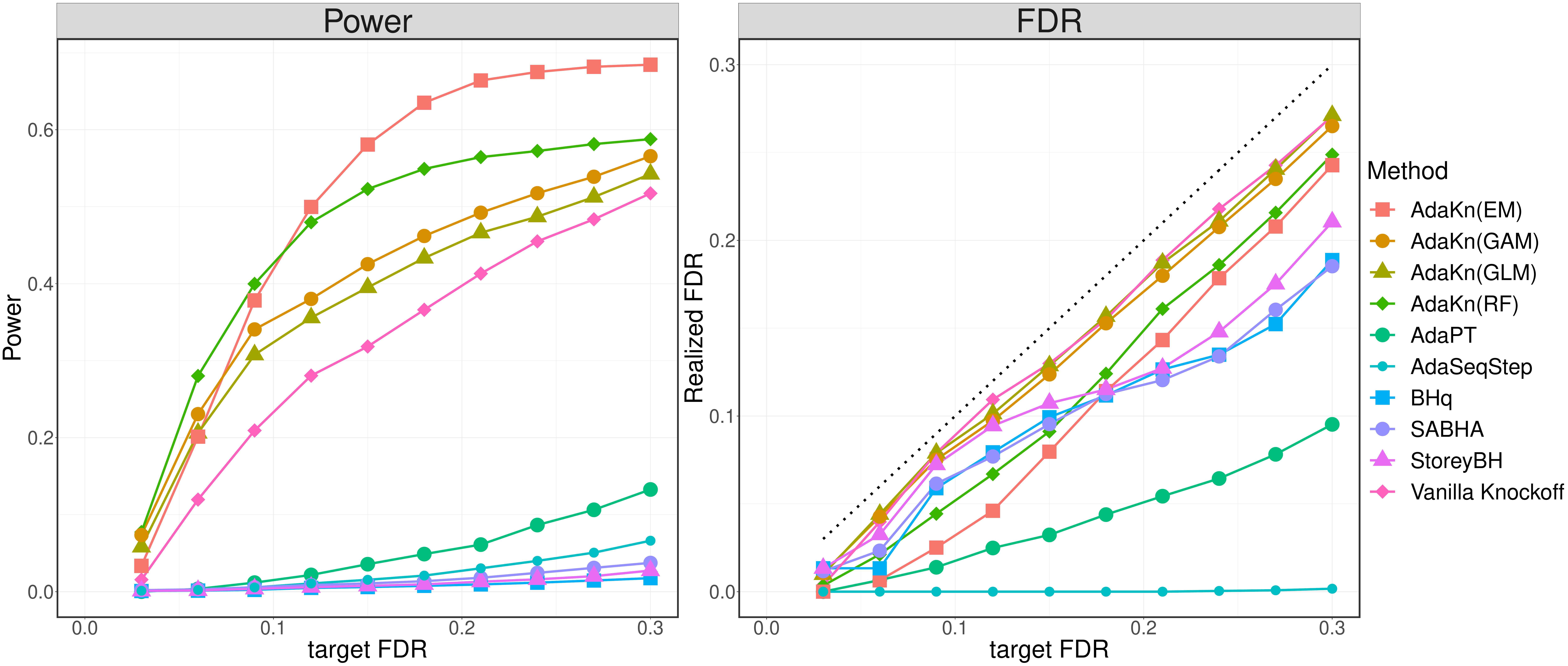}
\caption{Power (left) and FDR (right) versus target FDR values.
}
\label{sim1_res}
\end{figure}
\begin{figure}[h!]
\begin{subfigure}{0.45\textwidth}
\centering
\includegraphics[width = 0.9\textwidth]{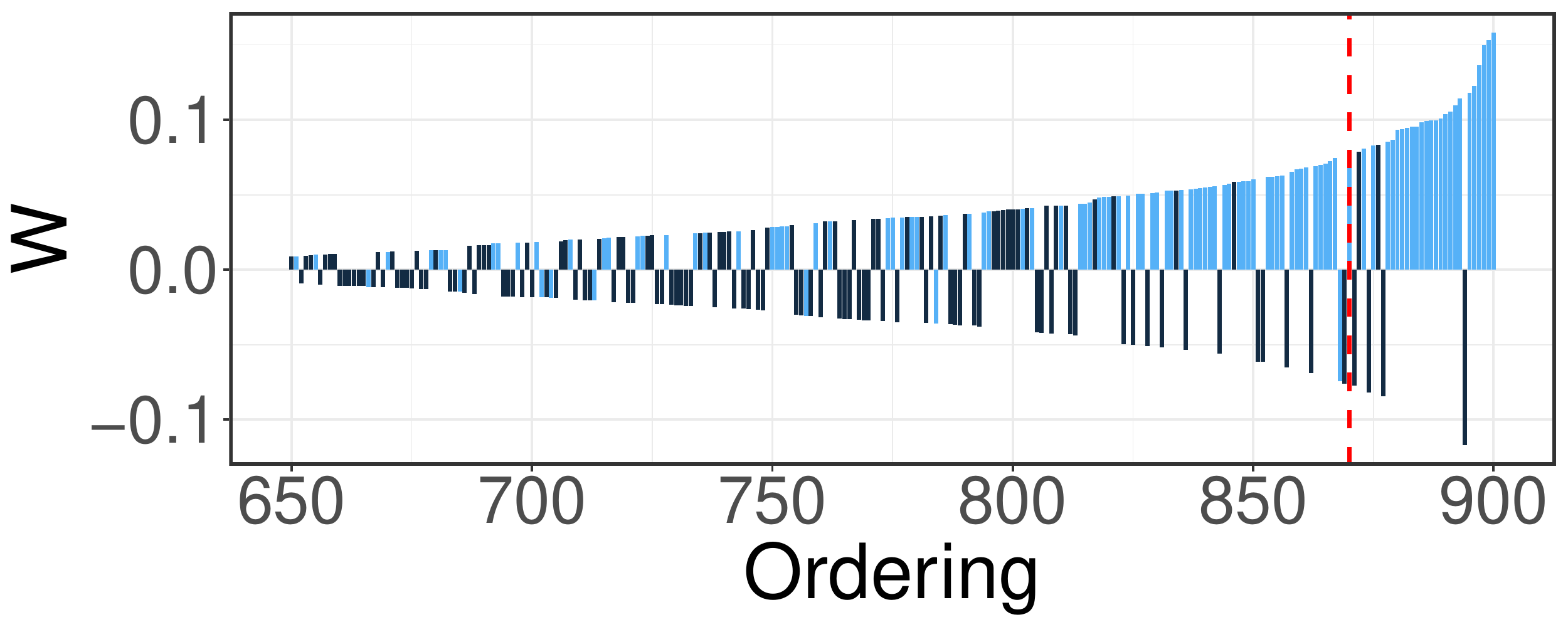}
\caption{Model-X knockoffs.}
\label{fig:vordering1}			
\end{subfigure}
\begin{subfigure}{0.45\textwidth}
\centering
\includegraphics[width = 0.9\textwidth]{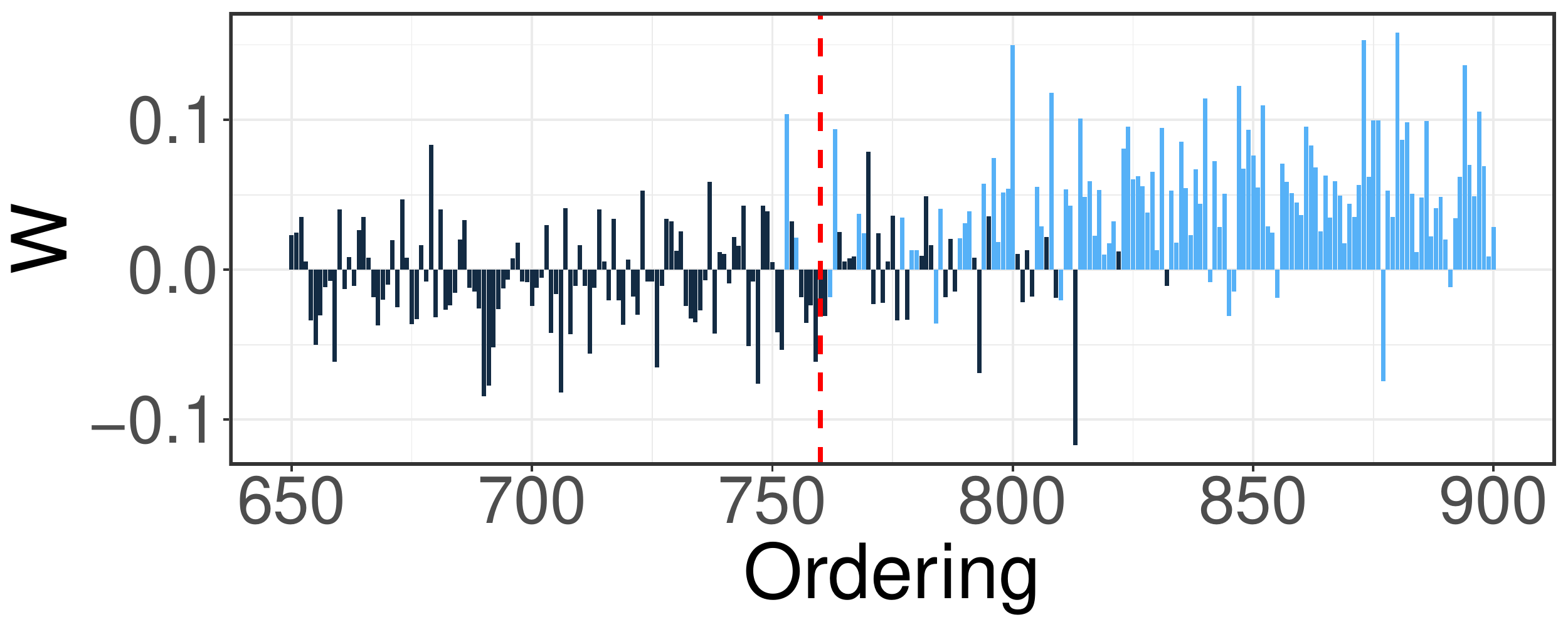}
\caption{Adaptive knockoffs.}
\label{fig:adaordering1}			
\end{subfigure}
\caption{(a) Realized ordering of vanilla knockoffs. (b) Realized
  ordering of adaptive knockoffs with the Bayesian filter. The x-axis
  is the ordering index and the y-axis is $W$. The blue bars represent
  the non-nulls; the black bars represent the nulls. The dashed red
  lines correspond to the selection thresholds for $q=0.2$, i.e., the
  features after the red line with positive signs are selected.  }
\label{fig:adaordering}			
\end{figure}

\subsection{Simulation 2: two-dimensional side information}\label{sec:simul2}

The simulated dataset is of size $n = 1000$ and $p=1600$. Conditional on $X$, $Y$ is generated from a logistic model:
\begin{align}
Y|X_1,\ldots,X_p \sim \text{Bernoulli}\left(\dfrac{\exp(\beta_1 X_1+\ldots +\beta_pX_p)}{1+\exp(\beta_1 X_1+\ldots +\beta_pX_p)}\right).
\end{align}
The entries of $\beta$ `live' on a two-dimensional plane and the
location of $\beta_j$ on the plane is described by a pair of
coordinates $(r(j),s(j))$, as in Figure \ref{fig:str}. In all, there are
$m=201$ blue nodes, representing the nonzero entries of
$\beta$. Details about the signal locations are in Appendix
\ref{sec:impdetails}. The magnitude of the nonzero entries is set to
$\frac{25}{\sqrt{n}}$ and the signs are generated via i.i.d.~coin
flips. The vector $X$ of covariates is drawn i.i.d.~from a
discrete-time Gaussian process with zero mean and covariance structure: 
\begin{align}
\cov(X_i,X_j) = e^{-3 ||U_i-U_j||_2^2},\qquad i,j\in[p],
\end{align}
where $U_j = (r(j),s(j))$. The side information is
the pair of coordinates of each feature.

This simulation setting is motivated by magnetic resonance imaging
(MRI) studies. For example, the hypotheses (nodes) are the voxels in a
structural MRI scan and the response is a $0$-$1$ variable indicating
whether the subject has Alzheimer's disease. Due to the spatial
correlation between the nodes, the signals often exhibit cluster
structures and our setup presents a simplified version of such
structures. Given the context, one may ask whether we should treat the
clusters themselves rather than the voxels as unit of inference. The
debate between cluster-based inference and voxel-based inference seems
still ongoing in the neuroimaging society. In particular, researchers
have recently observed that cluster-based inference often suffers from
low specificity (we do not know how many significant voxels there are
within a significant cluster) and neuroscientists are calling for
inference methods with higher resolution (see e.g.,
\citet{woo2014cluster,rosenblatt2018all}). Here we adopt the
voxel-based inference as in \citet{efron2012large}.

\begin{figure}[h!]
\centering
\includegraphics[width=0.5\textwidth]{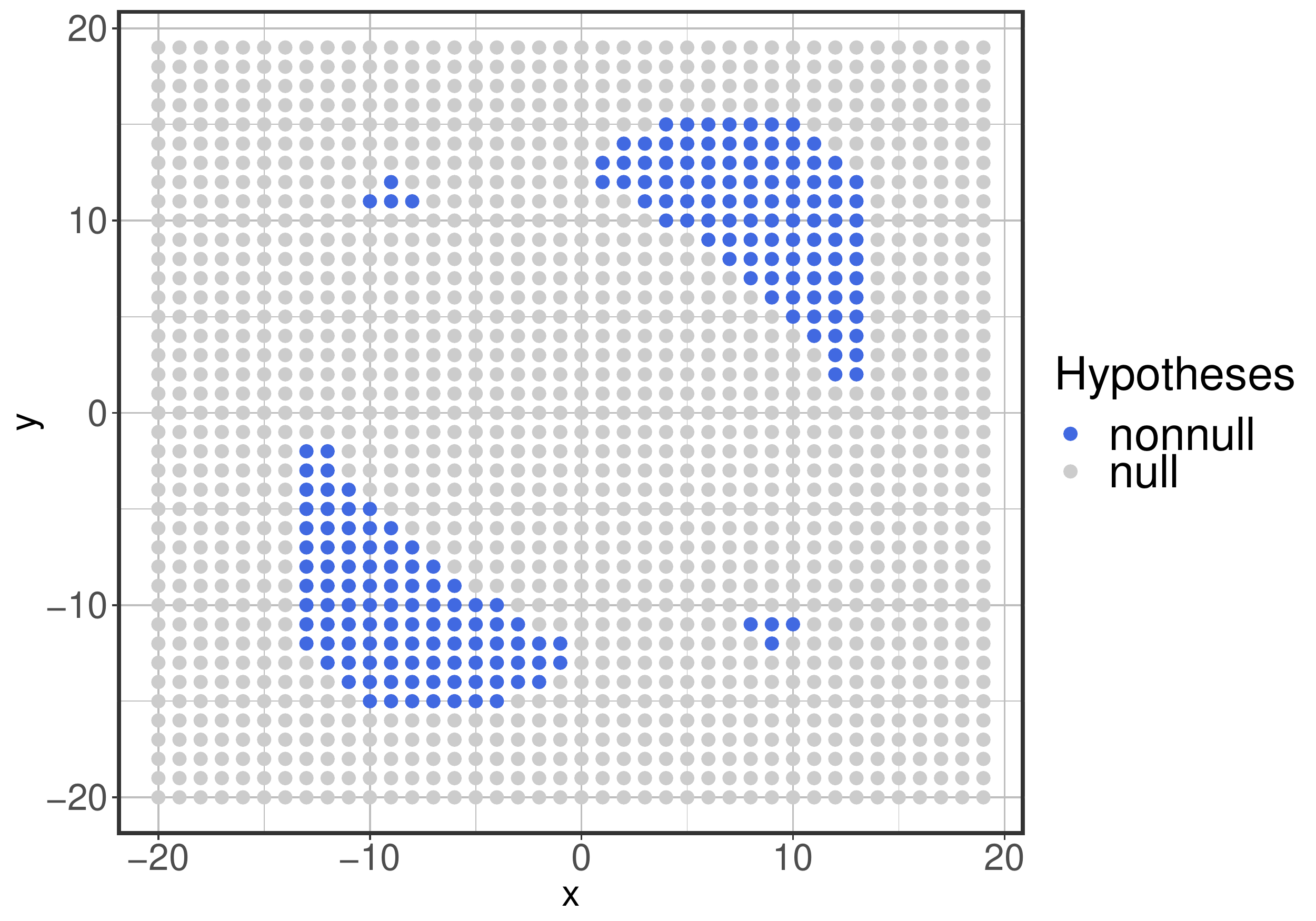}
\caption{Two-dimensional hypothesis structure. Blue nodes correspond to non-nulls and gray nodes to nulls. }
\label{fig:str}
\end{figure}

In this simulation, $p$ is larger than $n$, and obtaining valid
p-values is a problem. Hence, we here focus on comparing the
knockoff-based methods. Figure \ref{sim2_res} shows the power and FDR
of all the candidate methods. Again all methods control the FDR as
expected. Adaptive knockoffs with a Bayesian filter or random forest
filter outperform vanilla knockoffs by a wide margin. Figure
\ref{fig:vordering2} and \ref{fig:adaordering2} show the realized
ordering of vanilla knockoffs and adaptive knockoffs with the Bayesian
filter respectively. Adaptive knockoffs is able to place more
non-nulls towards the end of the ordering and has higher power. The GLM and GAM filters have almost the same power as vanilla knockoffs
because their models are too simple and cannot capture the
two-dimensional structure of the side information.

\begin{figure}[h!]
\centering
\includegraphics[width = 1\textwidth]{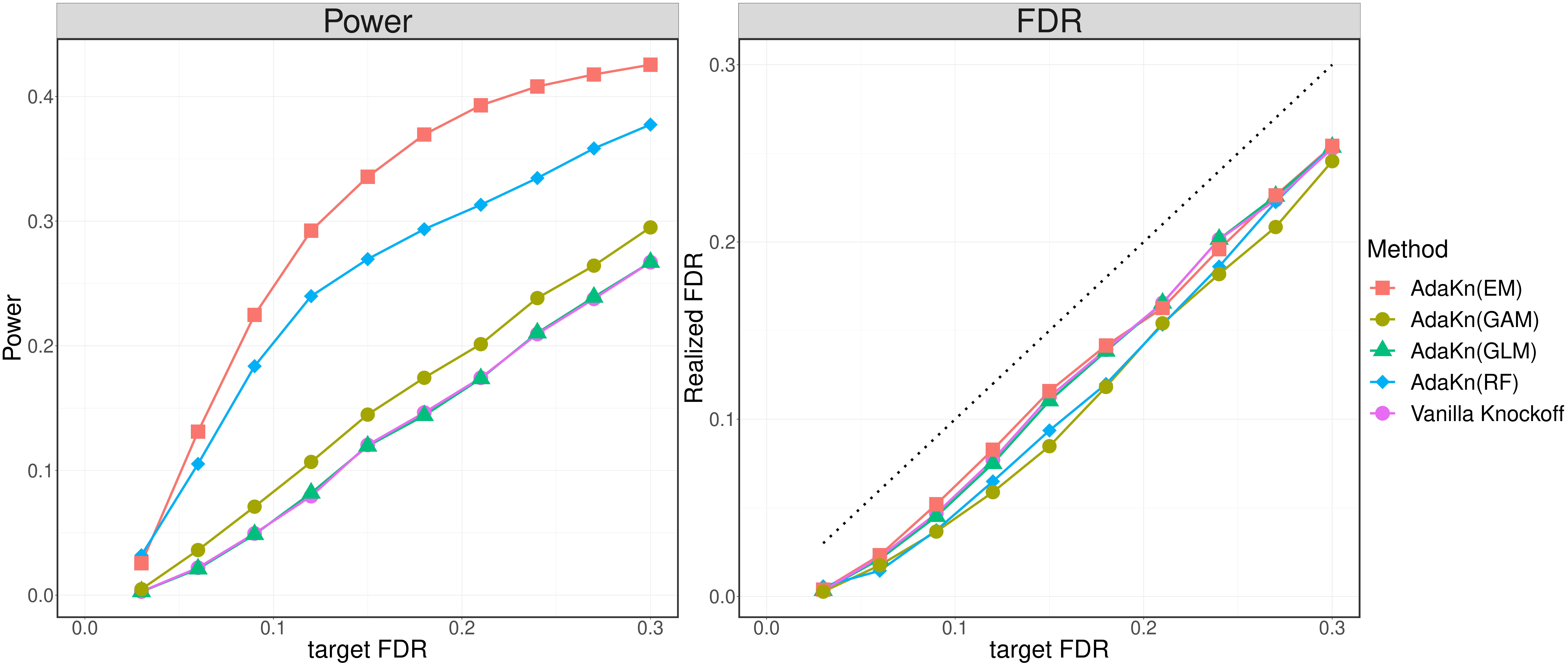}
\caption{Power (left) and FDR (right) versus target FDR values.}
\label{sim2_res}
\end{figure}

\begin{figure}[h!]
\begin{subfigure}{0.45\textwidth}
\centering
\includegraphics[width = 0.9\textwidth]{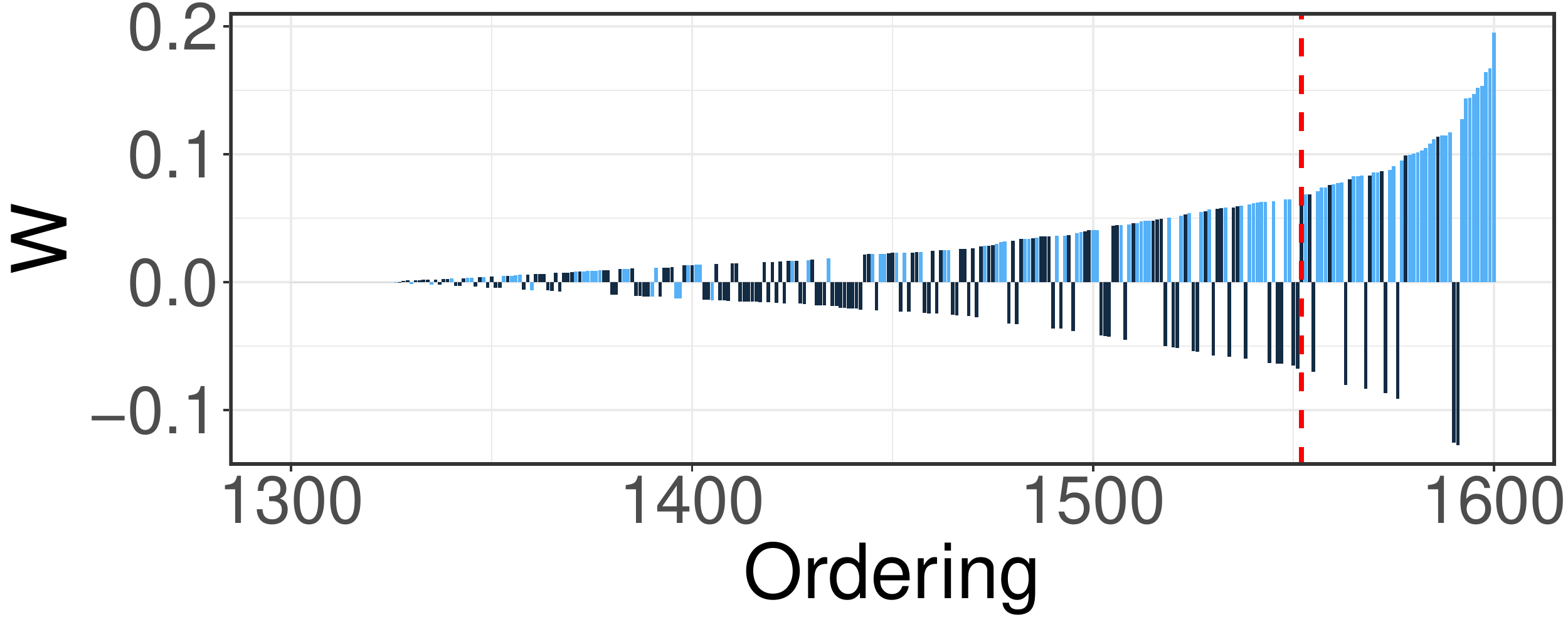}
\caption{Model-X knockoff.}
\label{fig:vordering2}			
\end{subfigure}
\begin{subfigure}{0.45\textwidth}
\centering
\includegraphics[width = 0.9\textwidth]{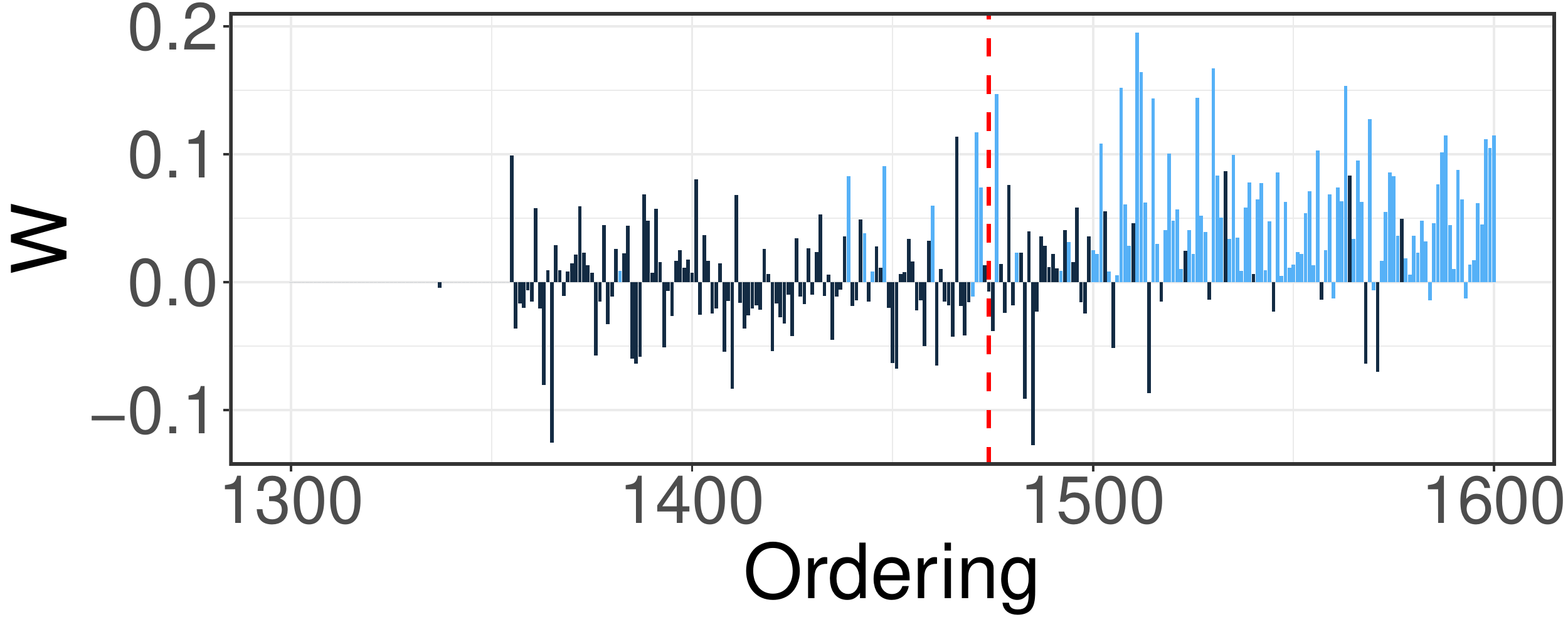}
\caption{Adaptive knockoffs.}
\label{fig:adaordering2}			
\end{subfigure}
\caption{(a) Realized ordering of vanilla knockoffs in Simulation 2; (b) Realized ordering of adaptive knockoffs with the Bayesian filter. The setup is otherwise the same as in Figure \ref{fig:adaordering}.
}
\end{figure}

\section{Applications}\label{sec:application}

\subsection{GWAS}
In Section \ref{sec:motivatingexample} we have
presented the results of our method applied to the
WTCCC dataset (Crohn's disease). In this section, we
discuss in detail the data analysis implementation,
and in addition, apply our methods to the Northern
Finland 1996 Birth Cohort study of metabolic syndrome
(NFBC).
\subsubsection{Overview of the data}
We have already described the WTCCC dataset in Section
\ref{sec:motivatingexample}. The NFBC dataset contains
information on $n=5402$ individuals from northern
Finland that includes genotypes at approximately
$300,000$ SNPs and nine phenotypes. The exact number
of effective observations are slightly different
across phenotypes because values are missing in some
of them. In this paper, we focus on low-density
lipoprotein (LDL) and high-density lipoprotein (HDL)
phenotypes. The inferential goal is to discover SNPs
significantly associated with LDL and HDL in the
Finnish population.

\subsubsection{Data pre-processing and SNP pruning}
\paragraph{Pre-processing} For the WTCCC Crohn's disease dataset, we
follow the pre-processing steps in \citet{candes2018panning} and for
the NFBC dataset, we follow the pre-processing steps in
\citet{sabatti2009genome,barber2019knockoff,sesia2018gene}. Table
\ref{tab:datasetdescript} lists the number of SNPs left after
pre-processing in the column named ``$p$ (pre-clustering)''.

\paragraph{Clustering} After pre-processing we further conduct a
clustering step to deal with the high correlation between SNPs. We
follow the method in
\citet{candes2018panning,sesia2018gene,barber2019knockoff} to cluster
SNPs and choose a representative from each cluster. Table
\ref{tab:datasetdescript} lists the number of SNP clusters in the
column ``$p$ (post-clustering)''. From now on, our inferential goal is
to discover important SNP clusters.

\begin{table}[h!]
\centering
\begin{tabular}{|c|c|c|c|c|}
\hline
Dataset & Phenotype & $n$ & $p$ (pre-clustering)& $p$ (post-clustering)\\
\hline
WTCCC& CD & $4913$ & $377{,}749$ & $71,145$\\
\hline
NFBC & LDL & $4682$ & $328{,}934$ & $59,005$\\
\hline
NFBC & HDL & $4700$ & $328{,}934$ & $59,005$\\
\hline
\end{tabular}
\caption{Description of the datasets.}
\label{tab:datasetdescript}
\end{table}

\subsubsection{Side information acquisition}\label{sec:sideinfo}
\paragraph{Crohn's disease} As discussed in Section
\ref{sec:motivatingexample}, we obtain the marginal p-values from
inflammatory bowel disease (IBD) studies in East Asia and Belgium
\citep{franke2010genome,liu2015association,goyette2015high} as side
information. In case a SNP is recorded in both studies, we use a
  weighted mean of the p-values as the side information; we give a
  larger weight to the p-values from the East Asia study because it
  contains more samples (the weights are respectively $1 - 1/101$ and
  $1/101$). When a SNP in our dataset is not recorded in a study, we
  apply the procedure above after imputing the missing p-value with a
  one.

\paragraph{Lipids} For HDL and LDL, we obtain summary statistics
reported by \citet{loh2018mixed}.\footnote{The summary statistics are downloaded from \url{https://data.broadinstitute.org/alkesgroup/UKBB/}.} Their
results are based on the UK Biobank dataset, which comprises genetic
information on a range of phenotypes of individuals from the UK. The
genetic information in the UK population can serve as  a reference for our study in the Finnish population. Explicitly, we obtain the association p-values reported for ``self-reported high cholesterol
level''. As before, if no p-value is found to match a SNP, we set the
corresponding side information to be one. The motivation here is that if a SNP is not even recorded, the chance of being significant will likely be low.  

\subsubsection{Implementation details}
\paragraph{Knockoff construction} We use the HMM knockoffs from
\citet{sesia2018gene} and follow their suggestion to set the number of
latent haplotype clusters to twelve. Knockoffs are generated separately for $22$ chromosomes and for the two datasets.\\

\paragraph{Feature importance statistics} Given the response $Y$ and augmented normalized covariate matrix $(X,\tX)$, we perform Lasso regression of $Y$ on $(X,\tX)$ and obtain Lasso coefficients $(\beta,\tilde{\beta})$. The penalty parameter $\lambda$ is chosen from a $10$-fold cross validation. The resulting feature importance statistic for each SNP is the difference between the magnitude of the original and the knockoff Lasso coefficients, i.e., $W_j =|\beta_j| - |\tilde{\beta}_j|$.\\

\paragraph{Adaptive knockoff filter} Each SNP is associated with a
p-value obtained from other studies.  We do not directly feed the
p-values to our filter but instead order the SNPs according to their
p-values and use the ranks of the SNPs as input of our filter. We use
the Bayesian two-group model filter introduced in Section
\ref{sec:bayesfilter} with the default setting except for the fact
that in the initialization step, we reveal the features whose $|W_j|$
is below a pre-specified threshold.  The threshold is $0.03$ for
Crohn's disease, $0.005$ for LDL and $0.0005$ for HDL. As far as
estimating the FDR, we use the less conservative $\widehat{\fdr}_0$.

\subsubsection{Results}
We apply adaptive knockoffs with target FDR level $q=0.1$. Since the
knockoff-based algorithms are essentially random and depend on the
realizations of $\tX$, we generate $50$ knockoffs independently
conditioning on $(X,Y)$. We conduct analysis on every realization of
$\tX$ and report the average number of discoveries.  In Appendix
\ref{appendix:boxplot}, we provide boxplots of discovery numbers and
in Appendix \ref{appendix:snps} the full list of discovered SNPs. The
average number of discovered SNPs for Crohn's disease has been
presented in Table \ref{tab:discories}, and the results for the NFBC
dataset are shown in Table \ref{tab:nfbc}. We compare our results with
\citet{sabatti2009genome} and \citet{sesia2018gene}; the former adopts
a marginal test with a p-value threshold of $5\times 10^{-7}$, and the
latter adopts a $0.1$ target FDR level. The results show that our
algorithm greatly improves the power of the original knockoff
procedure.

We would like to stress once more that FDR control holds regardless of
the correctness of the p-values we use.  Also, we are not merely
re-discovering what is already known since side information concerns
other populations.

\begin{table}[h!]
\centering
\begin{tabular}{|c|c|c|}
\hline
Method& Number of discoveries (HDL) & Number of discoveries (LDL)\\
\hline
\citet{sabatti2009genome} & $5$ & $6$\\
\hline
HMM knockoffs \citep{sesia2018gene} & $8$ & $9.8$\\
\hline
\textbf{Adaptive knockoffs} & \textbf{12.5} & \textbf{18.3}\\
\hline
\end{tabular}
\caption{(Average) number of SNP discoveries
made by different methods with target FDR
level $q=0.1$. For knockoff-based
algorithms, the reported number is averaged
over multiple realizations of
$\tX$ (for adaptive knockoffs, this number
is $50$).}
\label{tab:nfbc}
\end{table}

%
%

\section{Future work}
This paper generalizes the knockoff procedure to a setting where side
information associated with features is available. We close by
discussing a few interesting directions for future work.

\paragraph{GWAS in the minority populations} 
In this paper, we applied the adaptive knockoff procedure to GWAS in
the British and Finnish population and obtained summary statistics
from other populations. It will be interesting to apply our method in
a setting where the inferential target is a minority population (e.g.,
African-Americans or Hispanic-Americans). In truth, minority
populations are often under-represented in GWAS and these studies are,
therefore, often underpowered. Since there are abundant genetic data
from the European population, exploiting information from this
population to empower GWAS in minority populations is becoming a
popular research topic (see e.g.,
\cite{coram2015leveraging,coram2017leveraging}).  Our method is
tantalizing because we have seen how easily we can use GWAS
statistics from one population to boost power in another.

\paragraph{Beyond GWAS}
Knockoff-based procedures have been successfully applied to genetics,
and we would like to see them used in other areas. One potential area
is neuroimaging, a field in which researchers are interested in
discovering locations in the brain that are associated with certain
trauma. The neuroimage data (e.g., structural MRI) also has a spatial
structure which can be used as side information. Applying adaptive
knockoffs to such datasets promises important diagnostic information.

{\small
\subsection*{Acknowledgement}
E.J.C.~was partially supported by the National Science Foundation via
grant DMS--1712800, by the Simons Foundation via the Math + X award,
and by a generous gift from TwoSigma.  Z.~R.~was partially supported
by the same Math + X award. Z.~R.~thanks Stephen Bates, Nikolaos Ignatiadis, Eugene
Katsevich and Matteo Sesia for their valuable comments on this
project.
}

\bibliography{adaptiveKnockoff}
\bibliographystyle{apalike}
\clearpage
\appendix

\section{Conditional expectation (Algorithm \ref{algo.em})}\label{appendix:supp}
Suppose we are at step $k$ of Algorithm \ref{algo.akn} and step $s$ of
Algorithm \ref{algo.em}.  Recall that $\calG$ is the $\sigma$-field
generated by the current information. Since the distribution of $W_j$
is a mixture of a point mass at $0$ and an absolutely continuous
distribution, we need to treat the conditional probability and
conditional expectation differently depending on whether $W_j = 0$ or
not. To avoid complication, we slightly abuse notation and let
$\p (W_j|\cdot)$ refer to a probability when $W_j = 0$ and to a
density otherwise.

\paragraph{Revealed hypotheses} For $j
\in\{\pi_1,\ldots,\pi_k\}$, the value of $W_j$
is known conditional on $\calG$. 
\begin{itemize}
\item The conditional expectation of $H_j$ is
\begin{align}	
\E[H_j|\calG] &= \dfrac{\p(H_j = 1,W_j |\calG)}{\p(W_j|\calG)},
\end{align}
where
\begin{align}
&\p(H_j = 1,W_j  |\calG) = \nu(U_j)p_1(W_j;U_j),\label{eqn:condprob1}\\
&\p(W_j|\calG)=\nu(U_j)p_1(W_j;U_j)+(1-\nu(U_j))p_0(W_j;U_j)\label{eqn:condprob2}.
\end{align}

\item The conditional expectation of $Y_jH_j$ is 
\begin{align}
\E[Y_jH_j|\calG]  & = Y_j \E[H_j|\calG],
\end{align}				
where $\E[H_j|\calG]$ has been computed above.
\end{itemize}
\paragraph{Unrevealed hypotheses} For $j\in
[p]\backslash\{\pi_1,\ldots,\pi_k\}$, we know the magnitude 
of $W_j$ conditional on
$\calG$ but not its sign. 
\begin{itemize}
\item The conditional expectation of $H_j$ is 
\[
\E[H_j|\calG] = \dfrac{\p(H_j = 1,|W_j ||\calG)}{\p(|W_j||\calG)},
\]
where
\begin{align}
\p(H_j=1,|W_j||\calG)=& \nu(U_j)p_1(|W_j|;U_j)+\nu(U_j)p_1(-|W_j|;U_j)\\
\p(|W_j||\calG) =& \nu(U_j)p_1(|W_j|;U_j)+\nu(U_j)p_1(-|W_j|;U_j)+\\
&(1-\nu(U_j))p_0(|W_j|;U_j)+(1-\nu(U_j))p_0(-|W_j|;U_j).
\end{align}

\item The conditional expectation of $Y_jH_j$ is 
\[\E[Y_jH_j|\calG] \\
= y_{j,1}\p(H_j=1,W_j>0|\calG)+ y_{j,2}\p(H_j=1,W_j<0|\calG),
\]
where
\begin{align}						
&\p(H_j=1,W_j>0|\calG) = \dfrac{\p(H_j=1,W_j = |W_j||\calG)}{\p(|W_j||\calG)},\label{eqn:condexp1}\\
&\p(H_j=1,W_j<0|\calG) = \dfrac{\p(H_j=1,W_j = -|W_j||\calG)}{\p(|W_j||\calG)},\label{eqn:condexp2}\\
& y_{j,1} = \log(\exp(|W_j|)+1)-|W_j|,\\
& y_{j,2} = \log(\exp(-|W_j|)+1)+|W_j|.
\end{align}
The numerators and denominators in \eqref{eqn:condexp1} and
\eqref{eqn:condexp2} have been calculated in \eqref{eqn:condprob1} and
\eqref{eqn:condprob2}.

\end{itemize}
\section{Boxplots of the number of discoveries in Section \ref{sec:application}}\label{appendix:boxplot}
We present boxplots of the number of discoveries from multiple knockoffs
realizations.

\begin{figure}[h!]
\begin{subfigure}{0.5\textwidth}
\centering
\includegraphics[width = 1\textwidth]{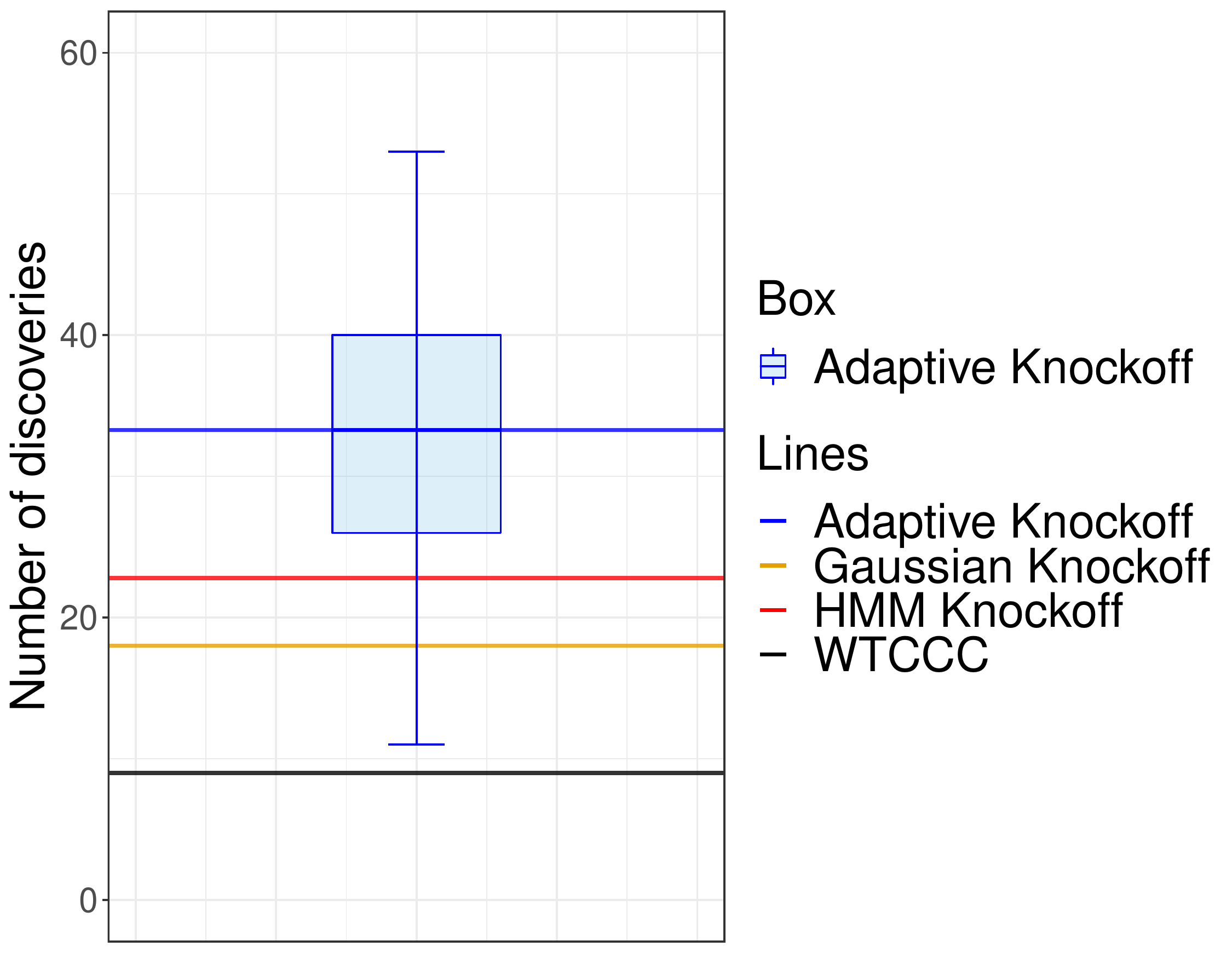}
\caption{Crohn's disease.}
\label{wtccc_rej}				
\end{subfigure}
\hfill
\begin{subfigure}{0.5\textwidth}
\centering
\includegraphics[width = 1\textwidth]{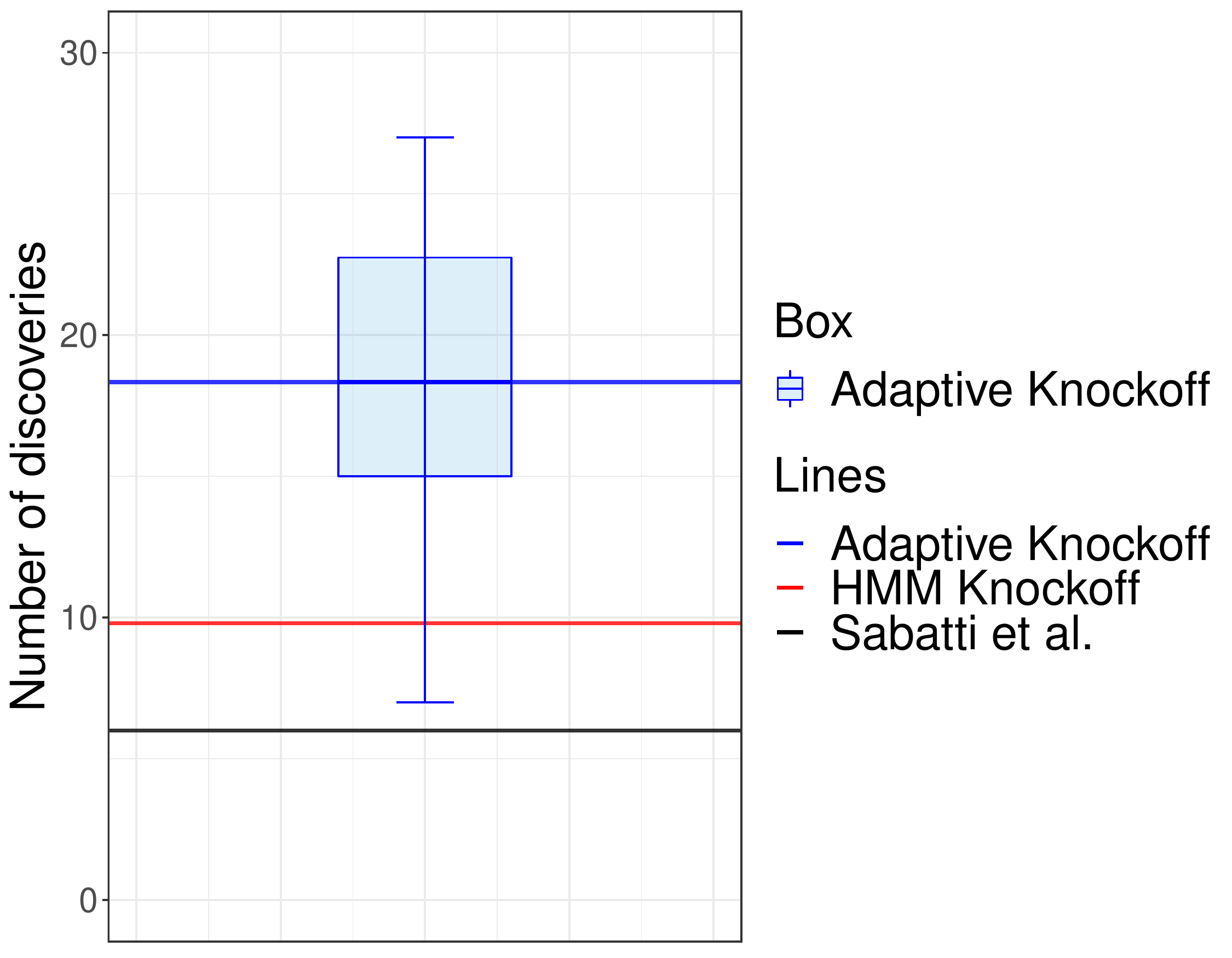}
\caption{LDL.}
\label{nfbc_ldl_rej}
\end{subfigure}
\hfill
\begin{subfigure}{0.5\textwidth}
\centering
\includegraphics[width = 1\textwidth]{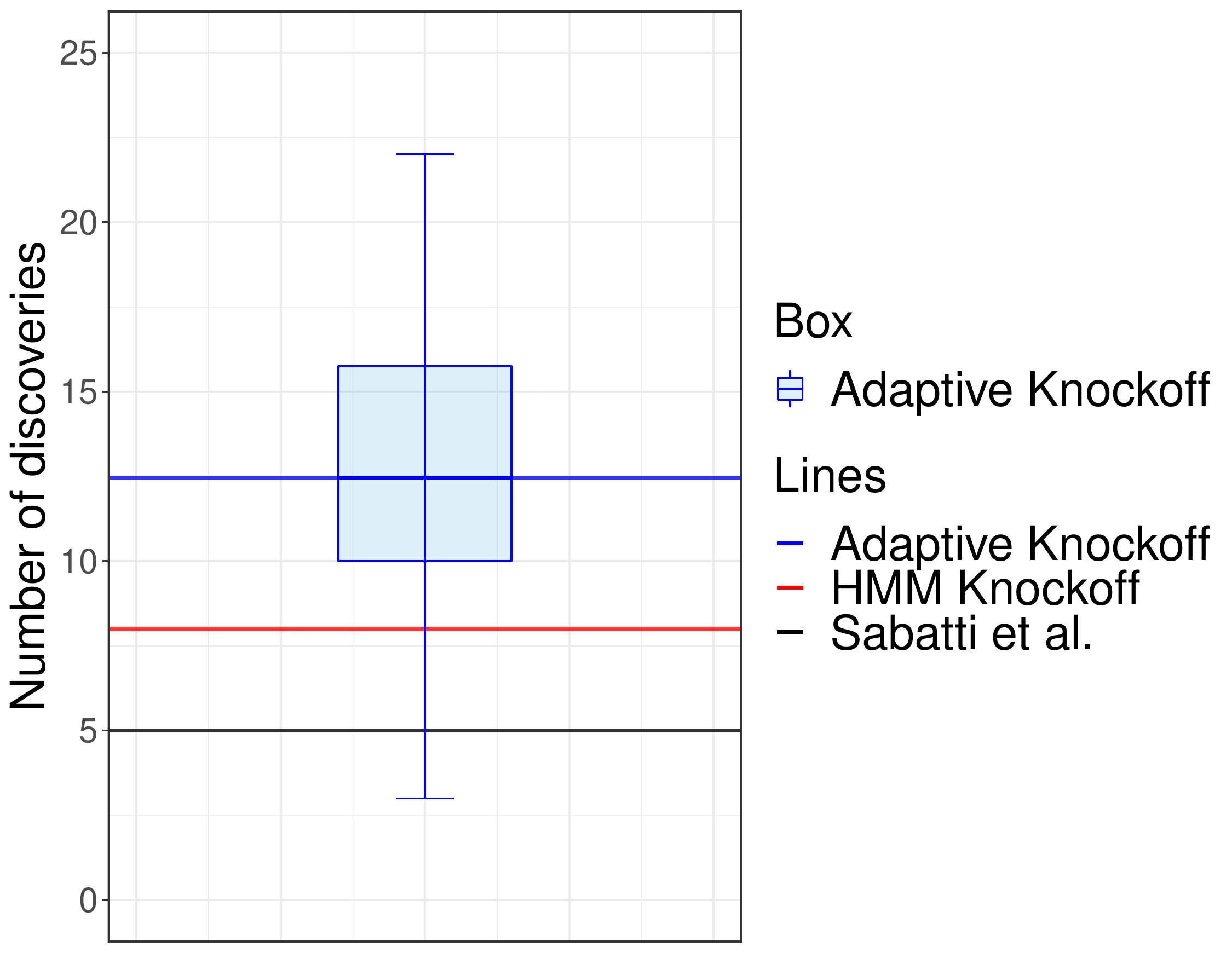}
\caption{HDL.}
\label{nfbc_hdl_rej}
\end{subfigure}
\caption{Number of discoveries from
  multiple knockoffs realizations.  The solid lines are the (average)
  numbers of discoveries and the boxplots represent the discoveries
  made by adaptive knockoffs in $50$ repetitions.}
\end{figure}
\clearpage
\section{Full list of discovered SNPs in Section \ref{sec:application}}\label{appendix:snps}
We present the full list of SNP clusters discovered by our adaptive
knockoff procedure. Since we run the algorithm $50$ times, we count
the frequency of SNP representatives being selected and report those
with a selection frequency greater than or equal to $30\%$. For each
cluster representative, we also report the size of the corresponding
cluster, the chromosome it belongs to, and the position range of the
cluster. The position of SNPs are reported as in the original dataset:
the WTCCC dataset follows the convention of Human Genome Build $35$
and the NFBC dataset follows the convention of Human Genome Build
$37$. Lastly, we compare our results to previous works. For Crohn's
disease, we indicate if our discovered SNPs are discovered by
\cite{wellcome2007genome,candes2018panning,sesia2018gene}. For NFBC we
compare with \cite{sabatti2009genome,sesia2018gene}.  An asterisk
indicates that the reported SNP is not exactly in the cluster but is
within the position range ($0.5$Mb).

\subsection{LDL}
\begin{table}[h!]
\centering
\begin{tabular}{|c|c|c|c|c|c|c|}
\hline
\makecell{Cluster\\ representative\\
(cluster size)} & \makecell{Selection\\ frequency (\%)} & Chr. & \makecell{Position\\ range (Mb)} & \makecell{Selection\\ frequency (\%) in\\  \cite{sesia2018gene}} & \makecell{Found in \\\cite{sabatti2009genome}?}\\
\hline
rs10198175 (1) &    100 & 2 & 21.13-21.13 & 80 & rs693∗\\
\hline
rs10953541 (58) &    100 & 7 & 106.48-107.30 & 76 & No.\\
\hline
rs157580 (4) &    100 & 19 & 45.40-45.41 & 94 & rs157580\\
\hline
rs2228671 (2) &    100 & 19 & 11.20-11.21& 97 & rs11668477\\
\hline
rs557435 (21) &    100 & 1 & 55.52-55.72 & 92 & No.\\
\hline
rs1713222 (45) &    98 & 2 & 21.11-21.53 & 41 & rs693\\
\hline
rs646776 (5) &    98 & 1 & 109.80-109.82 & 97 & rs646776\\
\hline
rs174450 (16) &    94 & 11 & 61.55-61.68 & 36 & rs1535\\
\hline
rs2802955 (1) &    92 & 1 & 235.02-235.02 & 40& No.\\
\hline
rs4803750 (1) &    90 & 19 & 45.25-45.25 & &  No.\\
\hline
rs6756629 (2) &    86 & 2 & 44.07-44.08 & & No.\\
\hline
rs688 (4) &    86 & 19 & 11.16-11.24 & & rs11668477*\\
\hline
rs4906908 (8) &    82 & 15 & 26.97-27.05 & & No.\\
\hline
rs12427378 (43) &    78 & 12 & 50.43-51.31 & 19 & No.\\
\hline
rs4844614 (34) &    74 & 1 & 207.30-207.88 & 99 & rs4844614\\
\hline
rs10409243 (5) &    70 & 19 & 10.33-10.37 & & No.\\
\hline
rs12670798 (11) &    68 & 7 & 21.57-21.71 & & rs693*\\
\hline
rs10056811 (94) &    60 & 5 & 74.24-75.24 & & No.\\
\hline
rs11878377 (39) &    50 & 19 & 10.63-11.18 & & rs646776*\\
\hline
rs11615 (17) &    44 & 19 & 45.91-46.10 & & No.\\
\hline
rs2919843 (3) &    42 & 19 & 45.19-45.20 & & No.\\
\hline
rs9696070 (6) &    40 & 9 & 89.21-89.24 & 25 & No.\\
\hline
rs1105879 (33) &    32 & 2 & 234.50-234.70 & & No.\\
\hline
\end{tabular}
\caption{SNPs clusters discovered to be associated with LDL.}
\label{tab:ldl_list}
\end{table}

\newpage
\subsection{HDL}
\begin{table}[h!]
\centering
\begin{tabular}{|c|c|c|c|c|c|c|}
\hline
\makecell{Cluster\\ representative\\
(cluster size)} & \makecell{Selection\\ frequency (\%)} & Chr. & \makecell{Position\\ range (Mb)} & \makecell{Selection\\ frequency (\%) in\\  \cite{sesia2018gene}} & \makecell{Found in \\\cite{sabatti2009genome}?}\\
\hline
rs1532085 (4) &    100 & 15 & 58.68-58.70 & 100 & rs1532085\\
\hline
rs1532624 (2) &    100 & 16 & 56.99-57.01 & 99 & rs3764261\\
\hline
rs1800961 (1) &    98 & 20 & 43.04-43.04 & 100 & No.\\
\hline
rs255049 (142) &    98 & 16 & 66.41-69.41 & 95 & rs255049\\
\hline
rs7499892 (1) &    98 & 16 & 57.01-57.01& 100 & rs3764261\\
\hline
rs10096633 (19) &    80 & 8 & 19.73-19.94 & 57 & No.\\
\hline
rs9898058 (1) &    78 & 17 & 47.82-47.82 & 55 & No.\\
\hline
rs17075255 (59) &    68 & 5 & 164.28-164.92 & 51 & No.\\
\hline
rs3761373 (1) &    62 & 21 & 42.87-42.87 & 43 & No.\\
\hline
rs12139970 (11) &    52 & 1 & 230.35-230.42 & 23 & No.\\
\hline
rs2575875 (10) &    52 & 9 & 107.63-107.68 & 28& No.\\
\hline
rs2849049 (6) &    44 & 9 & 15.29-15.31 && No.\\
\hline
rs173738 (3) &    42 & 5 & 16.71-16.73 & 12& No.\\
\hline
rs2019260 (24) &    38 & 5 & 16.41-16.59 && No.\\
\hline
rs2132167 (94) &    34 & 8 & 33.29-34.78 && No.\\
\hline
rs2426404 (1) &    32 & 20 & 50.64-50.64 && No.\\
\hline
rs9324799 (8) &    30 & 5 & 154.15-154.41 && No.\\
\hline
\end{tabular}
\caption{SNPs discovered to be associated with HDL.}
\label{tab:hdl_list}
\end{table}
\clearpage
\subsection{Crohn's disease}
\begin{longtable}[h!]{|c|c|c|c|c|c|c|c|c|}
\hline
\makecell{Cluster\\ representative\\
(cluster size)} & \makecell{Sel.\\ fre (\%)} & Chr. & \makecell{Position\\ range (Mb)} & \makecell{Selection\\ frequency (\%) in\\  Sesia et al.} & \makecell{Selection \\ frequency (\%) \\ in Cand{\`e}s et\\ al.?}& \makecell{Found in\\ WTCCC\\
et al.?}\\
\hline
rs11209026 (2) &    100 & 1 & 67.31-67.42 & 100 &100 & rs11805303*\\
\hline
rs11627513 (7) &    100 & 14 & 96.61-96.63 & 68 & 80 & No.\\
\hline
rs11805303 (16) &    100 & 1 & 67.31-67.46 & 95 & 80 & rs11805303\\
\hline
rs17234657 (1) &    100 & 5 & 40.44-40.44 & 97 & 90 & rs17234657\\
\hline
rs4246045 (46) &    100 & 5 & 150.07-150.41 & 66 & 50 & rs1000113\\
\hline
rs6431654 (20) &    100 & 2 & 233.94-234.11& 99 & 100 & rs10210302 \\
\hline
rs6500315 (4) &    100 & 16 & 49.03-49.07 & 73 & 60 & rs17221417\\
\hline
rs6688532 (33) &    100 & 1 & 169.40-169.65 & 98 & 90 & rs12037606 \\
\hline
rs7095491 (18) &    100 & 10 & 101.26-101.32 & 91 & 100 & rs10883365\\
\hline
rs2738758 (5) &    98 & 20 & 61.71-61.82 & 72& 60 & No.\\
\hline
rs4692386 (1) &    98 & 4 & 25.81-25.81& 56 & 40 & No.\\
\hline
rs3135503 (16) &    96 & 16 & 49.28-49.36 & 91 & 90& rs17221417\\
\hline
rs9469615 (2) &    94 & 6 & 33.91-33.92 & 48  & 30 & No.\\
\hline
rs4807569 (2) &    92 & 19 & 1.07-1.08 & 27 & & No.\\
\hline
rs6743984 (23) &    92 & 2 & 230.91-231.05 & 39 & 10 & No.\\
\hline
rs4263839 (23) &    90 & 9 & 114.58-114.78 & 56 & 30 & No.\\
\hline
rs17063661 (1) &    84 & 6 & 134.70-134.70 & & &No.\\
\hline
rs7497036 (19) &    82 & 15 & 72.49-72.73 & 22 & & No. \\
\hline
rs1451890 (26) &    78 & 15 & 30.92-31.01 & 15 & & No.\\
\hline
rs2390248 (13) &    78 & 7 & 19.80-19.89 & 54 & 50 & No.\\
\hline
rs10801047 (10) &    76 & 1 & 188.17-188.47 && &No.\\
\hline
rs1345022 (44) &    76 & 9 & 21.67-21.92& & 40 & No. \\
\hline
rs549104 (10) &    76 & 18 & 2.05-2.11 & & &No.\\
\hline
rs12529198 (31) &    74 & 6 & 5.01-5.10 & 23 & & No. \\
\hline
rs17694108 (1) &    74 & 19 & 38.42-38.42& &10&No. \\
\hline
rs10761659 (53) &    72 & 10 & 64.06-64.41 & 45 & 10& rs10761659\\
\hline
rs6601764 (1) &    72 & 10 & 3.85-3.85& 80& 100 & rs6601764 \\
\hline
rs7655059 (5) &    70 & 4 & 89.50-89.53 & 75& 40& No.\\
\hline
rs4870943 (10) &    66 & 8 & 126.59-126.62 & 14 & & No.\\
\hline
rs7768538 (1145) &    66 & 6 & 25.19-32.91 & 81& 60 & rs9469220\\
\hline
rs10172295 (127) &    48 & 2 & 57.86-59.01 & & & No.\\
\hline
rs946227 (4) &    48 & 6 & 138.12-138.13& & & No. \\
\hline
rs2836753 (5) &    44 & 21 & 39.21-39.23& 42 & 30 & No. \\
\hline
rs4959830 (11) &    42 & 6 & 3.36-3.41 & 20 & 10 & No.\\
\hline
rs9783122 (234) &    38 & 10 & 106.43-107.61 & 62 & 80 & No.\\
\hline
rs11579874 (17) &    34 & 1 & 197.61-197.76 &&&No.\\
\hline
rs4437159 (4) &    34 & 3 & 84.80-84.81 & 49 & 60 & No.\\
\hline
rs7726744 (46) &    34 & 5 & 40.35-40.71& 70& 50& rs17234657 \\
\hline
rs10916631 (14) &    30 & 1 & 220.87-221.08 & &40 & No.\\
\hline
rs2814036 (5) &    30 & 1 & 163.94-164.07&14& & No. \\
\hline
%
\caption{SNPs discovered to be associated with Crohn's disease.}
\label{tab:cd_list}
\end{longtable}

\section{Implementation details}
\label{sec:impdetails}
In Section \ref{sec:simul2}, feature $j$ is a signal if any of the
following conditions holds:
\begin{align}
\begin{cases}
&  ((\dfrac{r(j)}{9})^2+(\dfrac{s(j)}{9})^2-2)^4-(\dfrac{r(j)}{9})^3(\dfrac{s(j)}{9})^5<0,\\
& ((\dfrac{r(j)}{9}+1)^2+(\dfrac{s(j)}{9}-\dfrac{5}{4})^2-0.015<0,\\
& ((\dfrac{r(j)}{9}-1)^2+(\dfrac{s(j)}{9}+\dfrac{5}{4})^2-0.015<0.
\end{cases}
\end{align}

\end{document}